\documentclass[11pt,a4paper]{article}%

\usepackage{amsmath}%
\usepackage{amsfonts}%
\usepackage{amssymb}%
\usepackage{graphicx}%
\usepackage{graphics}%
\usepackage[sort]{natbib}%
\usepackage{breqn}%
\usepackage{setspace}
\usepackage{amsthm}
\usepackage{mathtools}
\usepackage[top=1.7in, bottom=1.7in, left=0.9in, right=0.9in]{geometry}
\usepackage{bigints}
\usepackage{epstopdf}
\usepackage{subfigure}

\newcommand{\keywords}[1]{\par\addvspace\baselineskip\noindent{\em Some key words}:\enspace\ignorespaces#1}

\title{Objective Bayesian modelling of insurance risks with the skewed Student-$t$ distribution}
\date{}
\author{Fabrizio Leisen, Juan-Miguel Marin and Cristiano Villa}
\begin{document}

\newtheorem{example}{\textit{Example}}
\newtheorem{definition}{Definition}
\newtheorem{theorem}{Theorem}
\newtheorem{lemma}{Lemma}
\newtheorem{proposition}{Proposition}
\newtheorem{axiom}{Axiom}
\newtheorem{remark}{\textit{Remark}}


\maketitle

\begin{abstract}
Insurance risks data typically exhibit skewed behaviour. In this paper, we propose a Bayesian approach to capture the main features of these datasets.  This work extends the methodology introduced in \cite{Villa:Walker:2014a} by considering an extra parameter which captures the skewness of the data. In particular, a skewed Student-\textit{t} distribution is considered. Two datasets are analysed: the Danish fire losses and the US indemnity loss.  The analysis is carried with an objective Bayesian approach. For the discrete parameter representing the number of the degrees of freedom, we adopt a novel prior recently introduced in \cite{Villa:Walker:2014b}.  

\end{abstract}

\section{Introduction}\label{sc_intro}
In this paper we discuss the application of a skewed-$t$ model suitable to capture skewness and kurtosis commonly present in insurance risks data. In particular, we introduce a Bayesian approach based on minimal informative prior distributions \citep{Villa:Walker:2014b} to estimate the parameters of the asymmetric Student-$t$ distribution (AST) introduced in \cite{Fernandez:Steel:1998} and re-proposed in \cite{Zhu:Galb:2010}.

Insurance risks data, for example insurance losses, tend to present skewed behaviour \citep{Lane:2000, Vernic:2006} in almost any circumstance. Furthermore, losses related to catastrophes (e.g. earthquakes, exceptional floods, etc)
are also better modelled by distributions with relatively heavy tails, as they quite often include extreme events. It is therefore appropriate to use a statistical model that can simultaneously account for data that is skewed and with values that may be extreme in behaviour. The generalised AST distribution, presented in Section \ref{sc_prelim}, includes a skewness parameter and allows to control the heaviness of both the left and right tails by means of a parameter representing the number of degrees of freedom. An obvious advantage of the AST, as it will be discussed, is that it allows to \textit{adjust} to the observed data, in the sense that can consider both skewed and non-skewed data and heavy-tailed data as well as observations that are suitable to be modelled by a distribution with tails like the one of the normal density. In other words, by estimating the parameter controlling the skewness of the distribution and the number of degrees of freedom, translates into a model selection scenario, where the competing models are: the normal distribution, the $t$ distribution, the skew normal and the skew $t$.

It is often important to consider cases where the prior information about the true parameter values is minimal. There are circumstances where the prior information is not available or, for some reasons, it is not reliable or practical to be used. As such, we consider an inferential scenario where the prior distributions for the parameters of the model are set up according to objective Bayesian criteria. Note that here, the term ``objective'' is not intended to represent an actually objective set up, but simply to categorise the procedure followed to derive the prior distributions as one that provides an output in a sort of an automated fashion \citep{Berger:2006}. The choice of the objective prior for the skewness parameter is the Jeffreys' prior, while for the location and the scale parameters we will adopt the reference prior. For the number of degrees of freedom we adopt the truncated discrete prior proposed in \cite{Villa:Walker:2014a}.

Although we do not discuss the possibility of using prior distributions elicited on the basis of reliable prior information, we still believe that it would be the most sensible way to proceed. Of course, the fact that inferential results can be effectively obtained by means of an objective Bayesian approach, it just strengthens the proposed model, as it allows to remove the (not always easy) task of translating prior information into prior probability distributions. \\

The paper is organised as follows. In Section \ref{sc_prelim} we introduce the asymmetric Student-$t$ distribution and discuss some of its properties. We also present the prior distributions used for the parameters and, in particular, the prior for the number of degrees of freedom, which we argue to have a support that is discrete and truncated. In Section \ref{sc_priors} we explicit the prior for the number of degrees of freedom and prove its independence from the skewness parameter. This result is key in motivating the prior for $\nu$ proposed by \cite{Villa:Walker:2014a}. Section \ref{sc_simul} is dedicated to present the results of a simulation study aimed to analyse the frequentist properties of the posterior distribution for the number of degrees of freedom. To complete the discussion of the proposed approach to model insurance loss data, in Section \ref{sc_real} we analyse two well known data sets. We are able to show the versatility of the model in a scenario of skewed data with extreme events and in a scenario where a symmetric normal distribution is sufficient to be used as a model. Final discussion points and conclusions are presented in Section \ref{sc_disc}.

\section{Preliminaries}\label{sc_prelim}

In this paper, we consider the asymmetric Student-$t$ (AST) distribution whose general density has the form
\begin{equation}\label{eq_ast1}
f(x|\alpha,\nu,\mu,\sigma) = \left\{
  \begin{array}{l l}
    \dfrac{K(\nu)}{\sigma}\left[1+\dfrac{1}{\nu}\left(\dfrac{x-\mu}{2\alpha\sigma}\right)^2\right]^{-\frac{\nu+1}{2}} & \quad x\leq\mu\\
    \dfrac{K(\nu)}{\sigma}\left[1+\dfrac{1}{\nu}\left(\dfrac{x-\mu}{2(1-\alpha)\sigma}\right)^2\right]^{-\frac{\nu+1}{2}} & \quad x>\mu
  \end{array} \right.
\end{equation}
where $\alpha\in(0,1)$ represents the skewness parameter, $\mu$ is the location parameter, $\sigma$ is the scale parameter, and $\nu$ represents the degrees of freedom, with
$$K(\nu)\equiv\Gamma((\nu+1)/2)/[\sqrt{\pi\nu}\Gamma(\nu/2)],$$
see \cite{Fernandez:Steel:1998} and \cite{Zhu:Galb:2010} for a detailed discussion of the distribution and its properties. The usual Student-$t$ distribution can be recovered by setting $\alpha=1/2$ in \eqref{eq_ast1}, while the skewed Cauchy and the skewed normal distributions are special cases obtained when $\nu=1$ and $\nu=+\infty$, respectively. 

In the Bayesian framework, the inference about the parameter of a model is accomplished by combining the prior uncertainty about the true value of the parameters, expressed in the form of probability distributions, and the information contained in the observed sample. The latter is expressed by the likelihood function which, for the model in \eqref{eq_ast1}, has the form
\begin{eqnarray}\label{eq_like1}
L(\alpha,\nu,\mu,\sigma|x) &=& \prod_{i=1}^n f(x_i|\alpha,\nu,\mu,\sigma) \nonumber \\
&=& \prod_{x_i\leq\mu}\dfrac{K(\nu)}{\sigma}\left[1+\dfrac{1}{\nu}\left(\dfrac{x_i-\mu}{2\alpha\sigma}\right)^2\right]^{-\frac{\nu+1}{2}} \times \prod_{x_i>\mu}\dfrac{K(\nu)}{\sigma}\left[1+\dfrac{1}{\nu}\left(\dfrac{x_i-\mu}{2(1-\alpha)\sigma}\right)^2\right]^{-\frac{\nu+1}{2}}.
\end{eqnarray}
Thus, if we indicate the joint prior distribution for the parameters by $\pi(\alpha,\nu,\mu,\sigma)$, the posterior distribution is given by
\begin{equation}\label{eq_prior_1}
\pi(\alpha,\nu,\mu,\sigma) \propto L(\alpha,\nu,\mu,\sigma|x)\times\pi(\alpha,\nu,\mu,\sigma).
\end{equation}
The prior distributions will be discussed in Section \ref{sc_priors}. However, we believe it is useful to give an overview of the overall approach in this section.

In this paper we adopt an objective Bayesian approach to make inference on the unknown parameters of the density in \eqref{eq_ast1}. To proceed, we assume that the tail parameter $\nu$ is discrete. That is, $\nu=1,2,\ldots$. The reason of choosing a discrete parameter space for $\nu$ is due to the close proximity of models with consecutive number of degrees of freedom. In line with what discussed for the $t$ density by \cite{Jacquier:2004} and \cite{Villa:Walker:2014a}, the amount of information from the observations (i.e. the sample size) will rarely be sufficient to discern between consecutive densities with a difference in the number of degrees of freedom smaller than one. As such, the choice of a discrete support for $\nu$, with intervals of size one, is sensible. It is still possible to define a discrete prior (with the same approach) over a more dense support; however, the choice has to be motivated by some prior evidence, such as the certainty of dealing with a very large sample size, and the resulting prior distribution will obviously be different as the amount of information would have changed.

\section{The prior distributions for the parameters of the AST model}\label{sc_priors}
In this section we outline the approach to derive the prior distributions for the parameters of the AST model. We dedicate most of this section to the prior for the number of degrees of freedom because, being a discrete parameter, presents some non-trivial challenges. We start by making the assumption that the parameters present, a priori, some degree of independence and, therefore, the prior has the form
\begin{equation}\label{eq_prior.1}
\pi(\alpha,\nu,\mu,\sigma) \propto \pi(\nu|\alpha,\mu,\sigma)\pi(\mu,\sigma)\pi(\alpha).
\end{equation}
As we will show in Section \ref{sc_nuprior}, the prior for the number of degrees of freedom does not depend on the skewness parameter, therefore $\pi(\nu|\alpha,\mu,\sigma)=\pi(\nu|\mu,\sigma)$. It has to be noted that, although objective methods aim to obtain prior distributions depending only on the chosen model, in practice there is always some degree of subjectivity involved. In this case, we make the assumption that the parameters are independent \emph{a priori}, noting that this is a common practice in objective Bayesian analysis.

\subsection{Objective Bayes for discrete parameter spaces}\label{sc_intro_obj}
One of the assumptions on which the model presented in this paper is based upon, is that the parameter representing the number of degrees of freedom ($\nu$) is considered as discrete. The specific reasons for this choice are presented in Section \ref{sc_nuprior}, but it is worthwhile to give an overview of the general idea behind the prior specific for $\nu$ here.

Whilst there are several approaches to deal with continuous parameter spaces, such as Jeffreys' prior \citep{Jeffreys:1961} and reference prior \citep{BBS:2009}, discrete parameter spaces have always been dealt with on a case-by-case basis. Only recently general solutions for the discrete case have been put forward, with different degree of success. Among these, possibly the most noteworthy include \cite{Barger:Bunge:2008} and \cite{BBS:2012}. However, the former approach can be applied to a limited set of models, and the latter approach shows some deficiencies in terms of objectivity and generality. In this paper, we consider the method discussed in \cite{Villa:Walker:2014b}, which, as far as we know, is the sole objective approach that can be applied to any discrete parameter space without the necessity of being ``adjusted'' to the chosen model.

The prior proposed in \cite{Villa:Walker:2014b} is based on the idea of assigning a \emph{worth} to each element $\theta$ of the discrete parameter space $\Theta$. The \emph{worth} is objectively measured by assessing what is lost if that parameter value is removed from $\Theta$, and it is the true one. Once the \emph{worth} has been determined, this will be linked to the prior probability by means of the self-information loss function \citep{Merhav:Feder:1998} $-\log\pi(\theta)$. A detailed illustration of the idea can be found in \cite{Villa:Walker:2014b}, but here is an overview.

Let us indicate by $f(x;\theta)$ a distribution (either a mass function or a density) characterised by the unknown discrete parameter(s) $\theta$ and let 
$$D(f(x;\theta)\|f(x;\theta^\prime))=\int f(x;\theta)\log\left\{\frac{f(x;\theta)}{f(x;\theta')}\right\}\,dx$$
be the Kullback--Leibler divergence \citep{Kull:1951}. The utility (i.e. \emph{worth)} to be assigned to $f(x;\theta)$ is a function of the Kullback--Leibler divergence measured from the model to the nearest one; where the nearest model is the one defined by $\theta^\prime\neq \theta$ such that $D(f(x;\theta)\|f(x;\theta^\prime))$ is minimised. In fact (see \cite{Berk:1966}) $\theta^\prime$ is where the posterior asymptotically accumulates if the true value $\theta$ is excluded from $\Theta$. The objectivity of how the utility of $f(x;\theta)$ is measured is obvious, as it depends on the choice of the model only.

Let us now write $u_1(\theta)=\log\pi(\theta)$ and let the minimum divergence from $f(x;\theta)$ be represented by $u_2(\theta)$. Note that $u_1(\theta)$ is the utility associated with the prior probability for model $f(x;\theta)$, and $u_2(\theta)$ is the utility in keeping $\theta$ in $\Theta$. We want $u_1(\theta)$ and $u_2(\theta)$ to be matching utility functions, as they are two different ways to measure the same utility in $\theta$. As it stands, $-\infty<u_1\leq0$ and $0\leq u_2<\infty$, while we actually want $u_1=-\infty$ when $u_2=0$. The scales are matched by taking exponential transformations; so $\exp(u_1)$ and $\exp(u_2)-1$ are on the same scale. Hence, we have
\begin{equation}\label{eq1}
e^{u_1(\theta)} = \pi(\theta) \propto e^{g\{u_2(\theta)\}},
\end{equation}
where
\begin{equation}\label{eq2}
g(u) = \log(e^u-1).
\end{equation}
By setting the functional form of $g$ in \eqref{eq1}, as it is defined in \eqref{eq2}, we derive the proposed objective prior for the discrete parameter $\theta$ as follows
\begin{equation}\label{eq3}
\pi(\theta) \propto \exp \left\{ \min_{\theta\neq \theta^\prime \in\Theta} D(f(x;\theta)\|f(x;\theta^\prime)) \right\} - 1.
\end{equation}
We note that in this way the Bayesian approach is conceptually consistent, as we update a prior utility assigned to $\theta$, through the application of Bayes theorem, to obtain the resulting posterior utility expressed by $\log\pi(\theta\mid x)$. Indeed, there is an elegant procedure akin to the Bayes Theorem which works from a utility point of view, namely that
$$\log\pi(\theta\mid x) = K + \log f(x\mid\theta) + \log\pi(\theta),$$
which has the interpretation of
$$\text{Utility}(\theta\mid x,\pi) = K + \text{Utility}(\theta\mid x) + \text{Utility}(\theta\mid\pi),$$
where $K$ is a constant which does not depend on $\theta$. There is then a retention of meaning between the prior and the posterior information (here represented as utilities). This property is not shared by the usual interpretation of Bayes theorem when priors are objectively obtained; in fact, the prior would usually be improper, hence not representing probabilities, whilst the posterior is (and has to be) a proper probability distribution.

\subsection{The prior for $\nu$}\label{sc_nuprior}

In this Section we will show that the prior for $\nu$ introduced in \cite{Villa:Walker:2014a} can be used as prior for $\nu$ for the skewed Student-$t$ distribution.  
The prior for $\nu$ is based on similar arguments as the ones discussed in \cite{Villa:Walker:2014a}. In particular, we consider the following.

First, the parameter is treated as discrete; that is, $\nu=1,2,\ldots$. The reason is that the amount of information provided from the data will rarely be sufficient to discern densities with $\nu$ separated by an interval smaller than one. In other words, in general, the sample size is not going to be sufficient to allow estimates more precise than size one.
Whilst it is possible, in principle, to consider any continuous parameter as discrete, for this model we deem appropriate to limit the discretisation to $\nu$. Besides the above motivation, we note that considering $\nu$ as discrete has connection with an interpretation of a $t$-distributed random variable. In fact, a $t$ with $\nu$ degrees of freedom can be seen as the ratio of two independent random variables: a standard normal and the square root of a chi-square divided by its number of degrees of freedom $\nu$. While in principle it is possible to discretise any continuous parameter, such as $\mu$, $\sigma$ and $\alpha$, the procedure would carries a strong degree of subjectivity, namely the discretisation density. In fact, the above considerations made for $\nu$ cannot be applied to the remaining three parameters of the model. Hence, the choice to consider $\nu$ only as discrete.


Second, the parameter space of $\nu$ has to be truncated. In \cite{Villa:Walker:2014a} this argument is motivated by the fact that, as the number of degrees of freedom ($\nu$) of a Student-$t$ goes to infinity, the density converges to a normal in distribution. As such, after a certain value of $\nu$, the model can be consider normal for any value of the parameter. A sensible choice is to set $\nu_{\max}=30$ to represent the normal model. Therefore, the inference problem reduces in choosing among $t$ densities with $\nu=1,\ldots,29$ and the normal density. As recalled in the introduction, a similar result holds for the skewed Student-$t$ distribution as well. Indeed, as $\nu$ goes to infinity, the model converges to a skewed normal distribution. This allows to apply the same truncation argument to the skewed Student-$t$ distribution.

To derive the prior for $\nu$, we apply the approach introduced in Section \ref{sc_intro_obj}. Let us, at first, assume that $\nu=1,\ldots,30$ and $\alpha=0.5$. In this case, the density in \eqref{eq_ast1} represents a symmetrical $t$ distribution with $\nu$ degrees of freedom, and the prior for the parameter $\nu$ has the form as in \cite{Villa:Walker:2014a}. To simplify the notation, we indicate the $t$-density with parameters $\nu$, $\alpha$, $\mu$ and $\sigma$ by $f_{\nu}^{\alpha}$; therefore, the prior for $\nu$ is given by
$$\pi(\nu|\alpha,\mu,\sigma) \propto \exp\left\{D(f_\nu^{\alpha}\|f_{\nu+1}^{\alpha})\right\}-1,$$
for $\nu<29$ and, for $\nu\geq29$
$$\pi(\nu|\alpha,\mu,\sigma) \propto \exp\left\{D(f_\nu^{\alpha}\|f_{\nu-1}^{\alpha})\right\}-1,$$
where $f_{30}^{\alpha}\sim \phi^{\alpha}$ is the skewed normal distribution with mean $\mu$ and standard deviation $\sigma$, i.e. 

\begin{equation}\label{eq_ast2}
\phi^{\alpha}(x|\mu,\sigma) = \left\{
  \begin{array}{l l}
    \dfrac{1}{\sigma\sqrt{2\pi}}\exp\left\lbrace -\left(\dfrac{x-\mu}{2\alpha\sigma}\right)^2\right\rbrace & \quad x\leq\mu\\
    \dfrac{1}{\sigma\sqrt{2\pi}}\exp\left\lbrace -\left(\dfrac{x-\mu}{2(1-\alpha)\sigma}\right)^2\right\rbrace & \quad x>\mu
  \end{array} \right.
\end{equation}
As recalled in the introduction, the choice of $\alpha=0.5$ corresponds to the usual Student-$t$ distribution and the prior above is the one introduced in 
\cite{Villa:Walker:2014a}. The following Theorem states a crucial result to set $\pi(\nu|\alpha,\mu,\sigma)$.
\begin{theorem}   Let $f_{\nu}^{\alpha}$ be the skewed Student-$t$ distribution with parameters $\mu$ and $\sigma$. Then 
$$D(f_{\nu}^{\alpha}\parallel f_{\nu+1}^{\alpha})=D(f_{\nu}^{0.5}\parallel f_{\nu+1}^{0.5}),$$ 
for every $\nu\geq 1$.
\end{theorem}
\begin{proof} Without loss of generality, we can consider $\mu=0$ and $\sigma=1$.
Note that
$$D(f_{\nu}^{\alpha}\parallel f_{\nu+1}^{\alpha})=D_{\leq}(f_{\nu}^{\alpha}\parallel f_{\nu+1}^{\alpha})+D_{>}(f_{\nu}^{\alpha}\parallel f_{\nu+1}^{\alpha}),$$
where
\begin{equation*}
\begin{split}
D_{\leq}(f_{\nu}^{\alpha}\parallel f_{\nu+1}^{\alpha})&=\int_{-\infty}^{0}f_{\nu}^{\alpha}(y)\log\left\{\frac{f_{\nu}^{\alpha}(y)}{f_{\nu+1}^{\alpha}(y)}\right\}\,dy,\\
\end{split}
\end{equation*}
and
\begin{equation*}
\begin{split}
D_{>}(f_{\nu}^{\alpha}\parallel f_{\nu+1}^{\alpha})&=\int_{0}^{+\infty}f_{\nu}^{\alpha}(y)\log\left\{\frac{f_{\nu}^{\alpha}(y)}{f_{\nu+1}^{\alpha}(y)}\right\}\,dy .\\
\end{split}
\end{equation*}
We focus on the first term 
\begin{equation*}
\begin{split}
D_{\leq}(f_{\nu}^{\alpha}\parallel f_{\nu+1}^{\alpha})&=\int_{-\infty}^0K(\nu)\left[1+\frac{1}{\nu}\left(\frac{y}{2\alpha}\right)^2\right]^{-\frac{\nu+1}{2}}\log\left\{\frac{K(\nu)\left[1+\frac{1}{\nu}\left(\frac{y}{2\alpha}\right)^2\right]^{-\frac{\nu+1}{2}}}{K(\nu+1)\left[1+\frac{1}{\nu+1}\left(\frac{y}{2\alpha}\right)^2\right]^{-\frac{\nu+2}{2}}}\right\}\,dy.\\
\end{split}
\end{equation*}
The change of variable $z={y}/{2\alpha}$ yields 
\begin{equation*}
\begin{split}
D_{\leq}(f_{\nu}^{\alpha}\parallel f_{\nu+1}^{\alpha})
&=2\alpha\int_{-\infty}^0 K(\nu)\left[1+\frac{z^2}{\nu}\right]^{-\frac{\nu+1}{2}}\log\left\{\frac{K(\nu)\left[1+\frac{z^2}{\nu}\right]^{-\frac{\nu+1}{2}}}{K(\nu+1)\left[1+\frac{z^2}{\nu+1}\right]^{-\frac{\nu+2}{2}}}\right\}\,dz\\
&=2\alpha D_{\leq}(f_{\nu}^{0.5}\parallel f_{\nu+1}^{0.5}).
\end{split}
\end{equation*}
In a similar fashion
$$D_{>}(f_{\nu}^{\alpha}\parallel f_{\nu+1}^{\alpha})=2(1-\alpha)D_{>}(f_{\nu}^{0.5}\parallel f_{\nu+1}^{0.5}).$$
The symmetry of the standard Student-$t$ distribution ensures that $$D_{\leq}(f_{\nu}^{0.5}\parallel f_{\nu+1}^{0.5})=D_{>}(f_{\nu}^{0.5}\parallel f_{\nu+1}^{0.5}).$$ 
Therefore 
$$2D_{\leq}(f_{\nu}^{0.5}\parallel f_{\nu+1}^{0.5})=2D_{>}(f_{\nu}^{0.5}\parallel f_{\nu+1}^{0.5})=D(f_{\nu}^{0.5}\parallel f_{\nu+1}^{0.5}),$$
and we can easily conclude 
\begin{equation*}
\begin{split}
D(f_{\nu}^{\alpha}\parallel f_{\nu+1}^{\alpha})&=D_{\leq}(f_{\nu}^{\alpha}\parallel f_{\nu+1}^{\alpha})+D_{>}(f_{\nu}^{\alpha}\parallel f_{\nu+1}^{\alpha})\\
&=2\alpha D_{\leq}(f_{\nu}^{0.5}\parallel f_{\nu+1}^{0.5})+2(1-\alpha)D_{>}(f_{\nu}^{0.5}\parallel f_{\nu+1}^{0.5})\\
&=\alpha D(f_{\nu}^{0.5}\parallel f_{\nu+1}^{0.5})+(1-\alpha)D(f_{\nu}^{0.5}\parallel f_{\nu+1}^{0.5})\\
&=D(f_{\nu}^{0.5}\parallel f_{\nu+1}^{0.5}).
\end{split}
\end{equation*}
\end{proof}
\noindent In a similar way, it can be proved that 
$$D(f_\nu^{\alpha}\|f_{\nu-1}^{\alpha})=D(f_\nu^{0.5}\|f_{\nu-1}^{0.5})$$
for every $\nu\geq 2$. The above result also holds when we assume that for $\nu=30$, $f_{\nu}^{\alpha}$ is the skewed normal distribution. These results lead to the following  important considerations:
\begin{enumerate}
\item The objective prior distribution for $\nu$ doesn't depend by the skewness parameter $\alpha$; 
\item The objective prior distribution for $\nu$ for the skewed model is exactly the prior introduced in \cite{Villa:Walker:2014a}, i.e. $\pi(\nu|\mu,\sigma,\alpha)=\pi(\nu|0.5,\mu,\sigma)$.
\end{enumerate}

\subsection{Prior distributions for $\alpha$ and $(\mu,\sigma)$}\label{sc_prior_others}
The derivation of non-informative priors for the remaining parameters of the AST is straightforward. A common assumption is that the parameters are independent a priori. Although this assumption can be relaxed, in the sense of limiting the independence to the one between the skewness parameter on one side and the location and scale parameters of the other side, the resulting overall prior is the same.

In fact, if we consider $\mu$ independent from $\sigma$, the Jeffreys' independent prior will have the form $\pi(\mu,\sigma)=\pi(\mu)\pi(\sigma)$. Given that the Jeffreys' prior for a location parameter is proportional to 1, and the Jeffreys' prior for a scale parameter is proportional to the inverse of the parameter, we would have $\pi(\mu,\sigma)\propto 1/\sigma$ \citep{Jeffreys:1961}. However, the above prior coincides with the reference prior for the pair $(\mu,\sigma)$ \citep{BBS:2009}. Therefore, assuming or not assuming prior independence between the location parameter and the scale parameter does not make any practical difference from an inferential point of view.

For the skewness parameter $\alpha$, we will use the Jeffreys' prior discussed in \cite{Rubio:Steel:2014}, which is a Beta distribution with both parameters equal to 1/2; that is $\pi(\alpha)\sim\mbox{Beta}(1/2,1/2)$. In fact, it can be seen in Theorem 3 and Corollary 3 of \cite{Rubio:Steel:2014} that the (independence) Jeffreys prior for $\alpha$,  under the parameterisation in \eqref{eq_ast1},  is precisely a Beta distribution  with both parameters equal to $1/2$.  

\section{Simulation study}\label{sc_simul}
In this section we present a simulation study of the prior for $\nu$. Due to the objective nature of the prior considered it is appropriate to present the frequentist properties of the yielded marginal posterior for the number of degrees of freedom. In particular, we analyse the frequentist mean squared error (MSE) and the frequentist coverage of the $95\%$ credible intervals. The former represents a measure of the precision of the estimate, while the latter reports the proportion of times the true value is contained in the interval defined by the $2.5\%$ and the $97.5\%$ quantiles of the posterior in a repeated sampling scenario. The posterior distribution for $\nu$ tends to be skewed, as for example is shown in Figure \ref{fig:single_samp}; as such, an appropriate index of the posterior which summarises its centrality is the median. Furthermore, considering $\nu$ as discrete calls for an index which is discrete as well, i.e. the median. Therefore, the MSE is computed with respect to the median, and the precision of the estimate is defined by the relative square root of the mean squared error from the median: $\sqrt{\mbox{MSE}(\nu)}/\nu$. Given that the model converges to the normal distribution for $\nu\rightarrow\infty$, it is more difficult to discern between AST densities with large values of $\nu$. By considering the relative MSE, we somehow counterbalance an otherwise naturally increasing MSE and give a more interpretable information about the performance of the prior. 

The posterior distribution for $\nu$ is obtained by marginalising the full posterior $\pi(\alpha,\nu,\mu,\sigma)\propto L(\alpha,\nu,\mu,\sigma|x)\pi(\alpha,\nu,\mu,\sigma)$, where $\pi(\alpha,\nu,\mu,\sigma)$ is the prior defined in \eqref{eq_prior.1} and $L(\alpha,\nu,\mu,\sigma|x)$ is the likelihood function defined in \eqref{eq_like1}. As the posterior distribution is analytically intractable, we use Monte Carlo methods to obtain the marginal posterior distributions for each parameter. The algorithm employed is outlined in the Appendix.

In this simulation study we consider the following scenarios. We noted that the location parameter and the scale parameter do not have any effect on the inferential results for the number of degrees of freedom, as such we have considered, without loss of generality, $\mu=0$ and $\sigma=1$. We have considered three different values of the skewness parameter, that is $\alpha=0.3$, $\alpha=0.5$ and $\alpha=0.8$. For each of the above three values we have performed repeated sampling with $\nu=1,\ldots,20$. We have run 100000 iterations of the MCMC algorithm for each case, and for a sequence of sample sizes $n=50$, $n=100$ and $n=1000$. The simulation has been done in  R language and it took about 5.4 hours per case in a cluster of computers with i7 processors. Although the main fields of application of the asymmetrical $t$ density discussed in this paper call for relatively large sample sizes, we have also considered the possibility of applying the model to small numbers of observations. It is in fact for relatively small sample size that Bayesian analysis tends to give better pay-off compared to the frequentist approach. Therefore, beside considering $n=1000$, we have also analysed the frequentist properties for $n=100$ and $n=50$. The prior used for the parameters $\alpha$, $\mu$ and $\sigma$ are the objective priors outlined in Section \ref{sc_prelim} which, as shown in \cite{Rubio:Steel:2014} and the references therein, yield proper posterior distributions.
\begin{figure}[h!]
\centering
\subfigure[]{%
\includegraphics[scale=0.25]{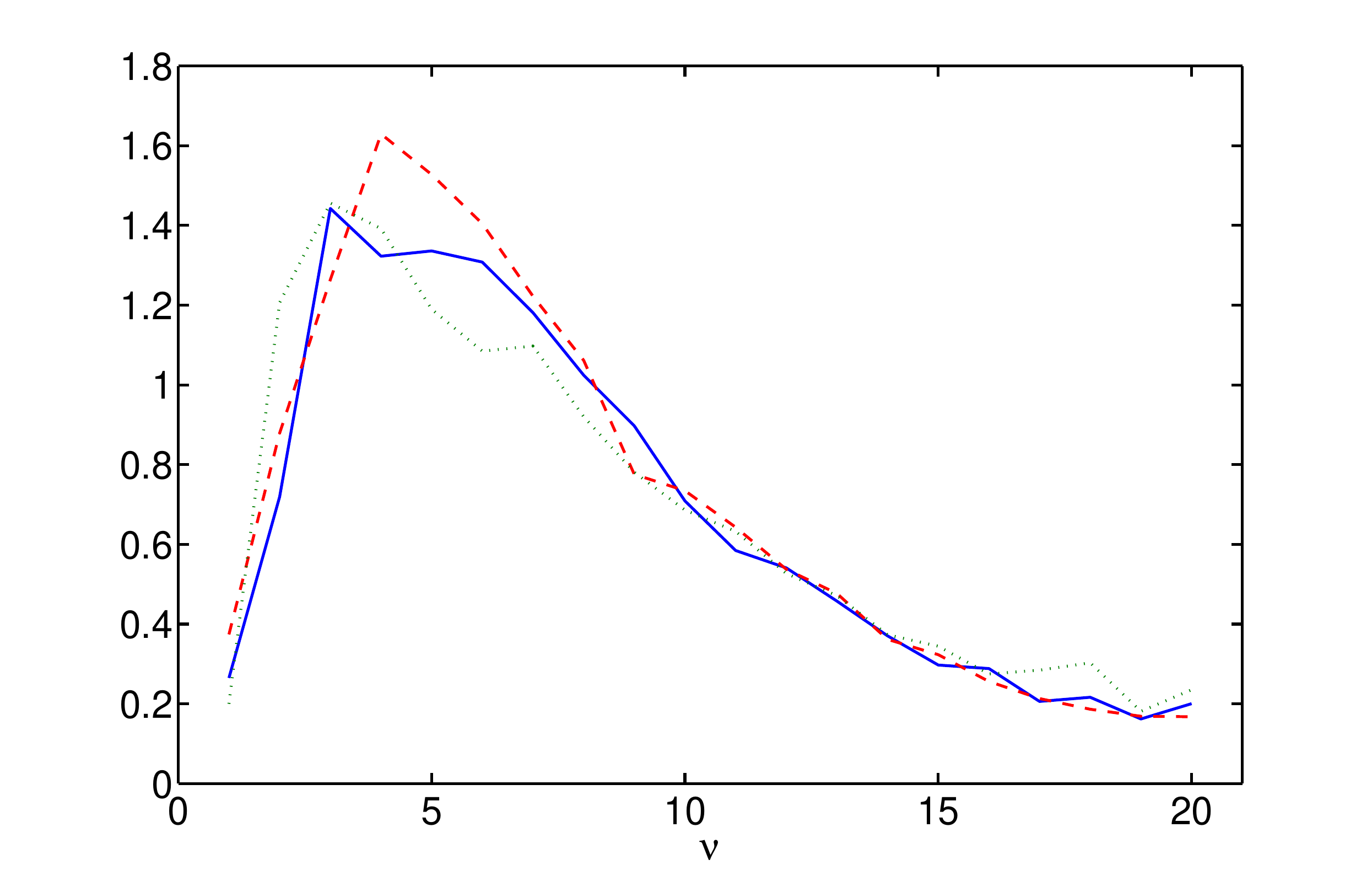} 
\label{fig:subfigure1}}
\quad
\subfigure[]{%
\includegraphics[scale=0.25]{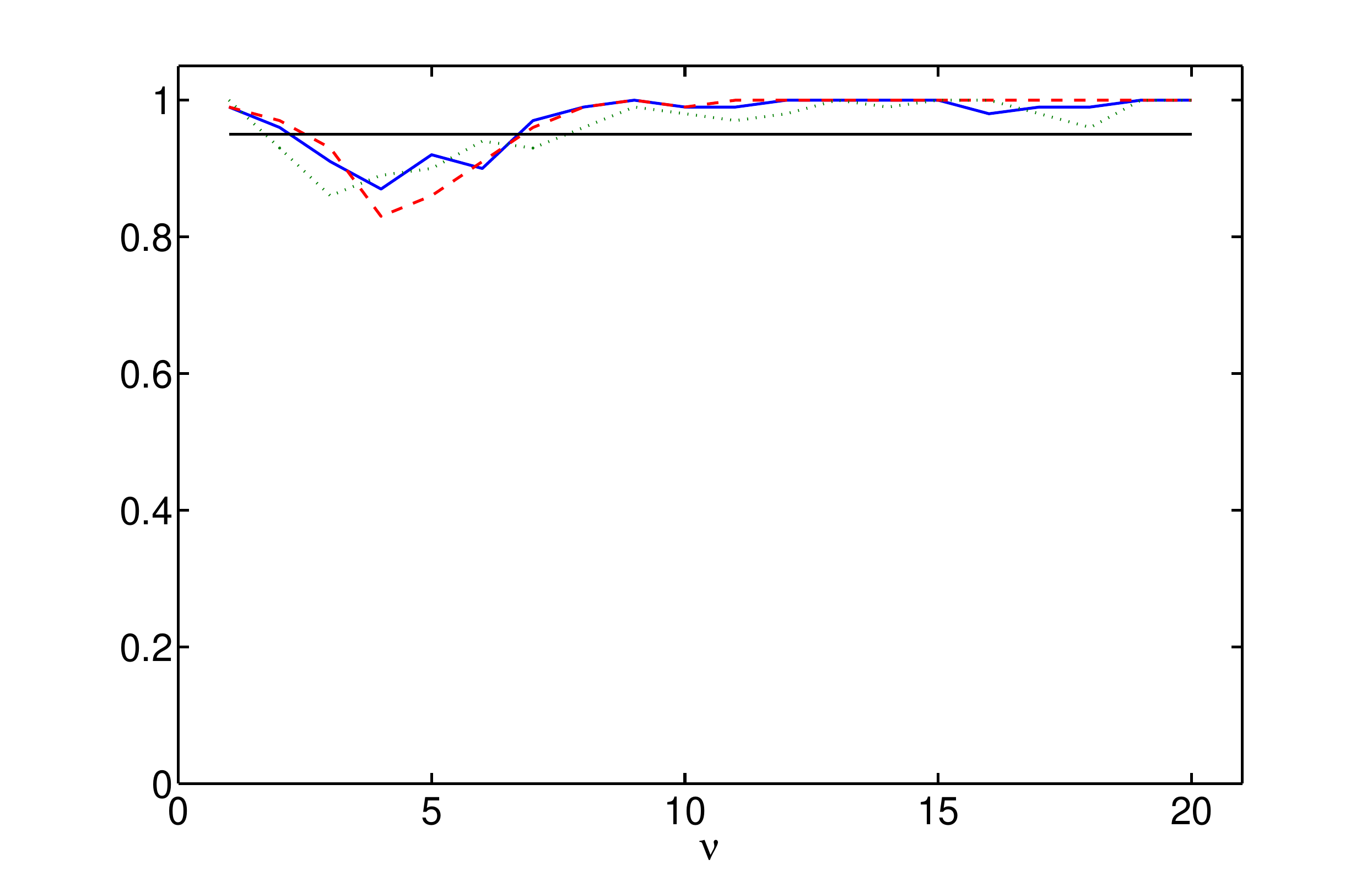} 
\label{fig:subfigure2}}
\subfigure[]{%
\includegraphics[scale=0.25]{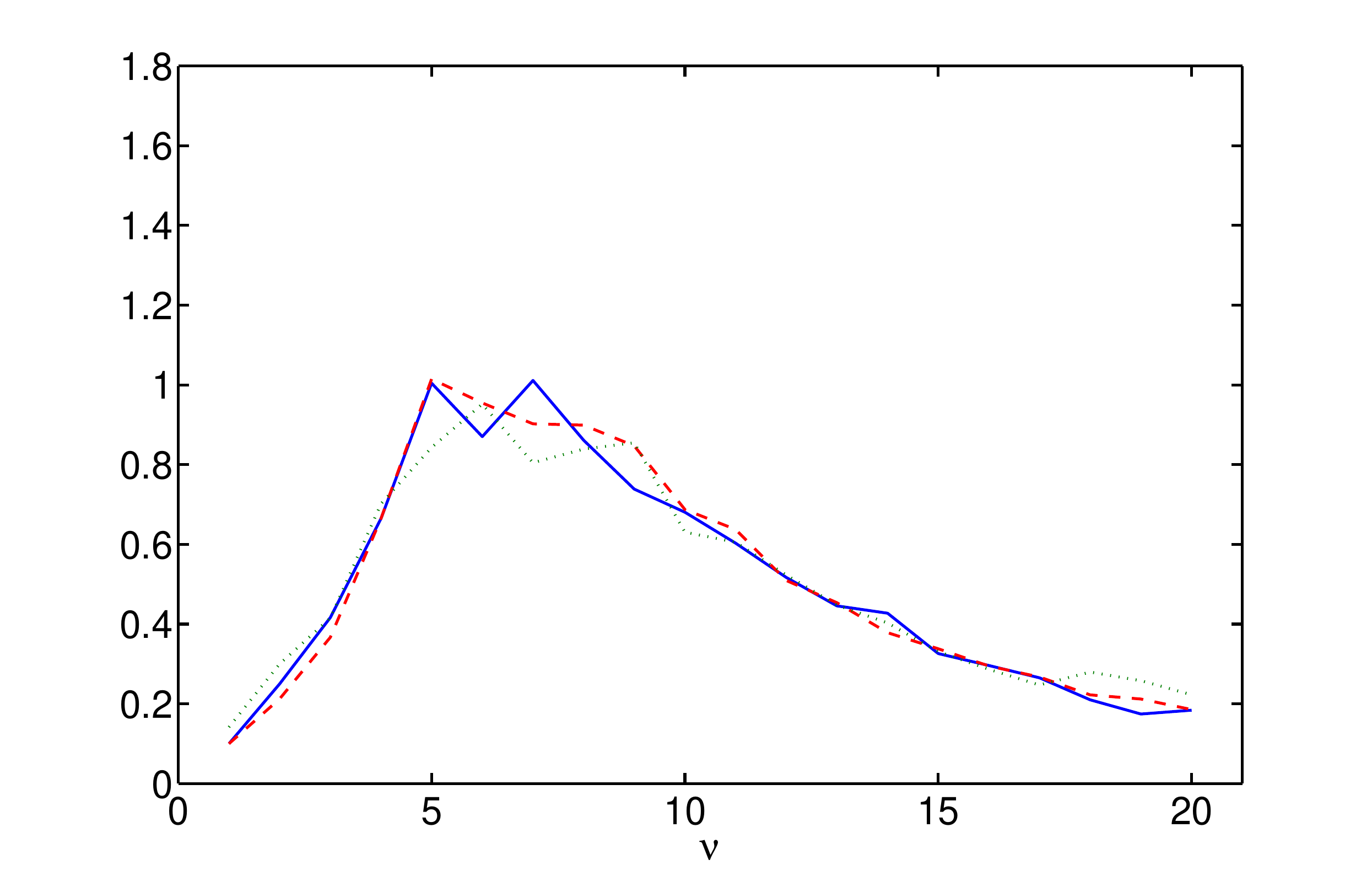}
\label{fig:subfigure3}}
\quad
\subfigure[]{%
\includegraphics[scale=0.25]{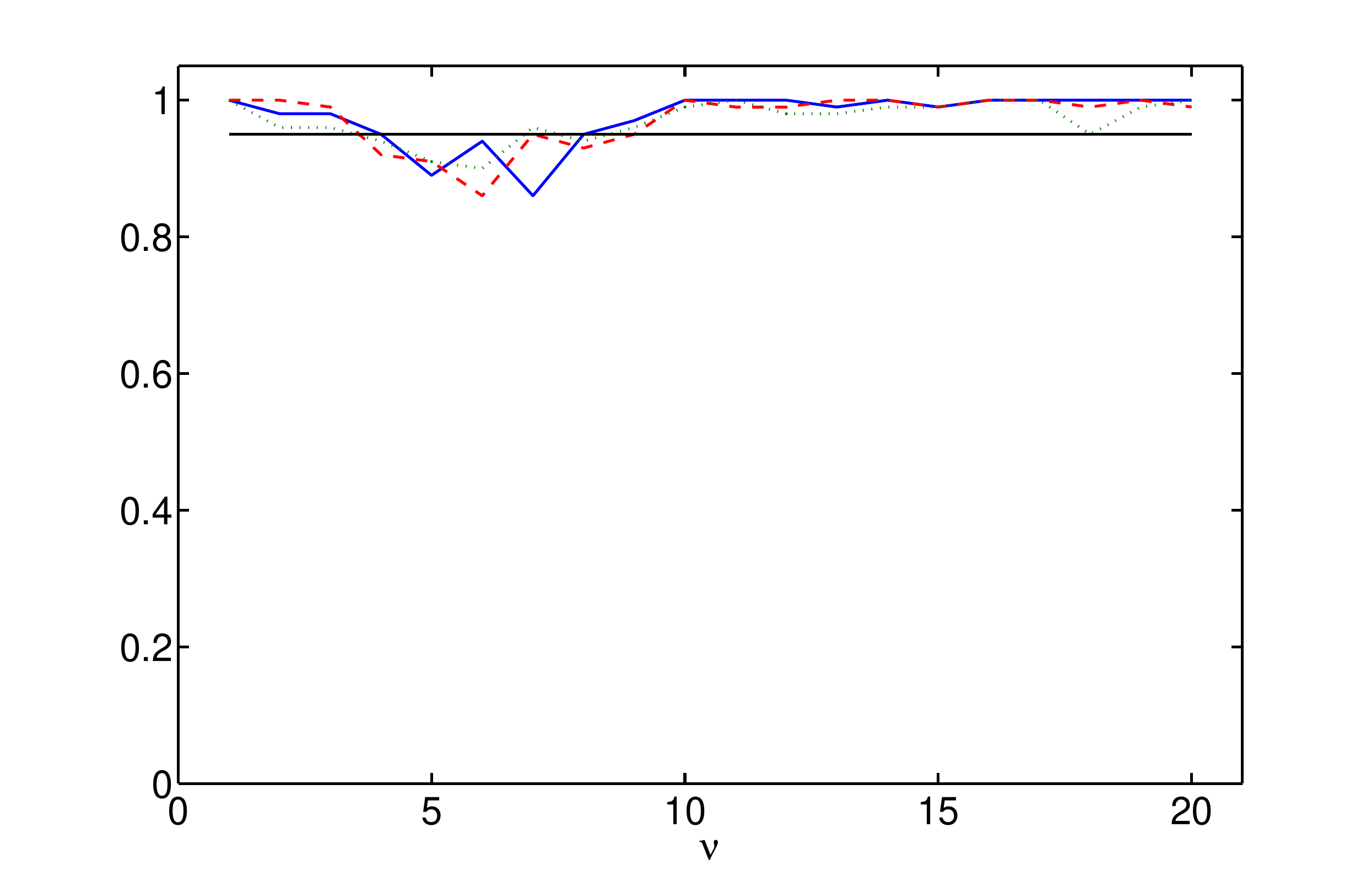}
\label{fig:subfigure4}}
\subfigure[]{%
\includegraphics[scale=0.25]{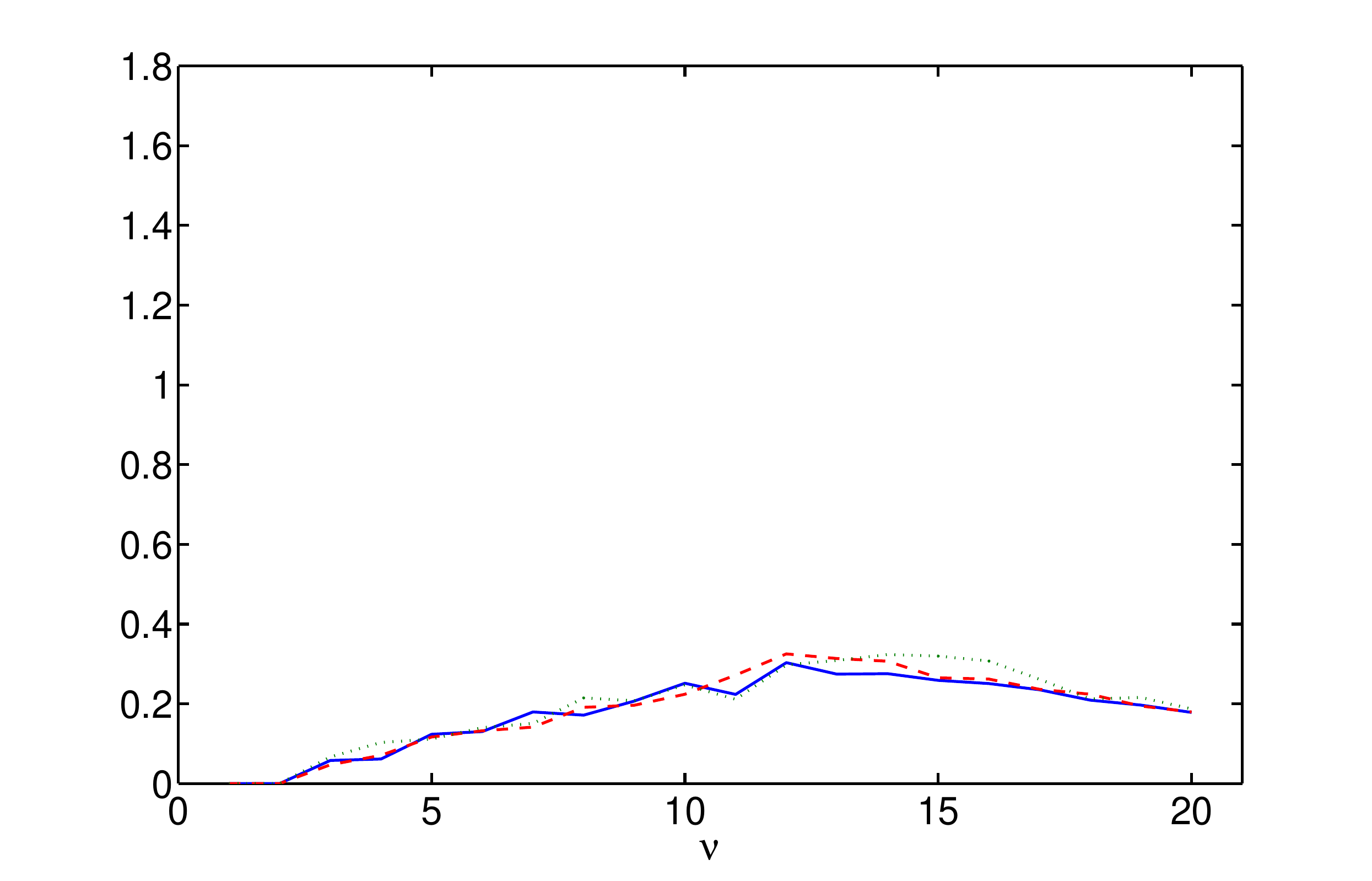}
\label{fig:subfigure5}}
\quad
\subfigure[]{%
\includegraphics[scale=0.25]{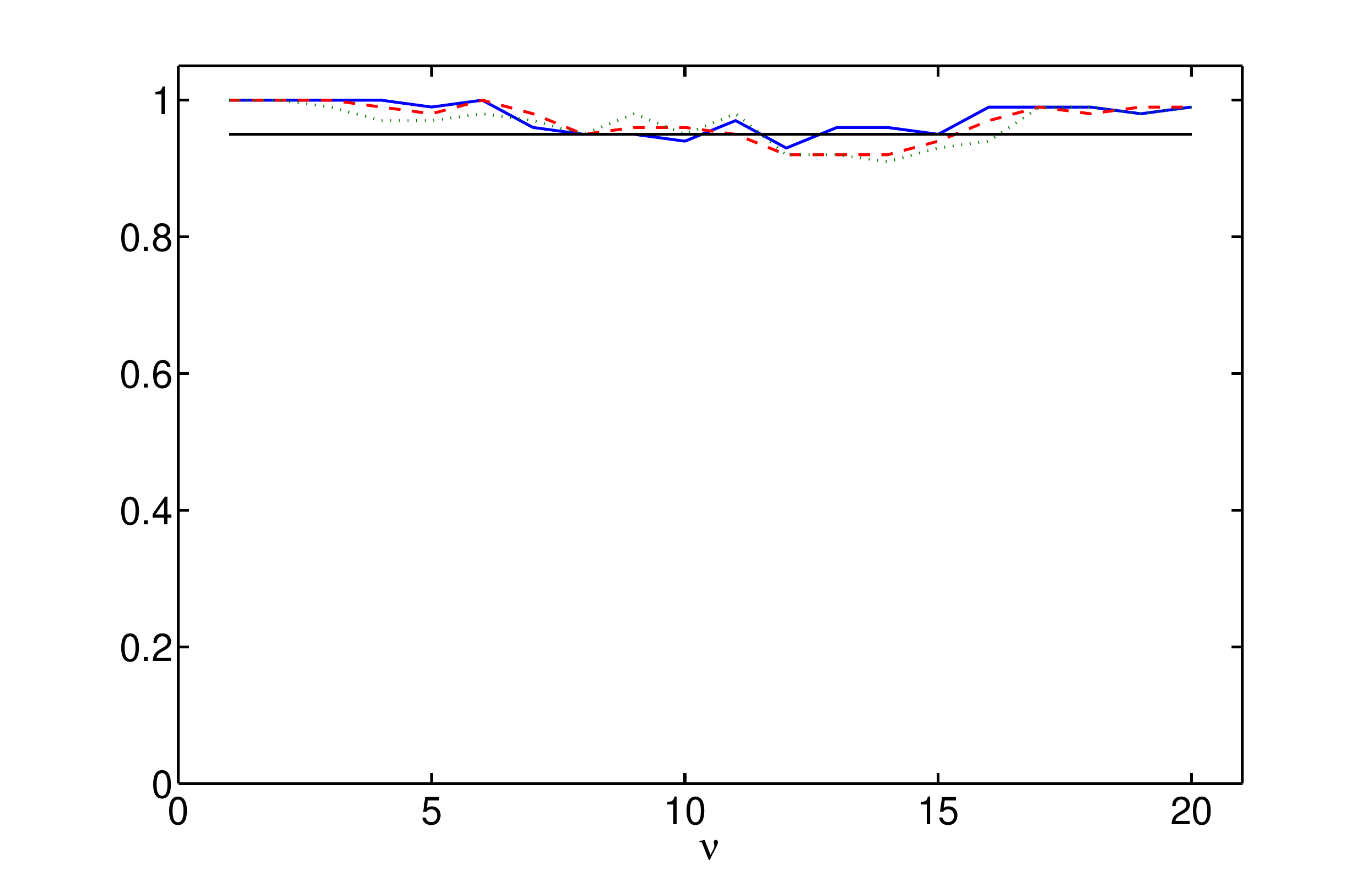}
\label{fig:subfigure6}}
\caption{Frequentist coverage of the $95\%$ credible intervals for $\nu$ (right) and square root of relative mean squared error of the estimator of $\nu$ (left). The simulations are for $\alpha=0.3$ (solid), $\alpha=0.5$ (dashed) and $\alpha=0.8$ (dotted), and for $n=50$ (top), $n=100$ (middle) and $n=1000$ (bottom).}
\label{fig:figure1}
\end{figure}
The simulation results are summarised in Figure \ref{fig:figure1}. The left column, plots (a), (c) and (e), shows the square root of the relative mean squared error for the posterior medians of $\nu$. While there is a dependence on the accuracy of the estimate from the sample size, we do not appreciate any effect from the value of the skewness parameter $\alpha$. In fact, within each plot on the left-hand-side, the mean squared error curves have a similar behaviour, with a higher value towards the region of the parameter space where contiguous number of degrees of freedom characterise distributions relatively different. As $\nu$ increases, leading to $t$ densities which are more and more similar to each other, the relative mean squared error decreases. As expected, the mean squared error in higher for small sample sizes than for larger sample sizes, given that more information is carried by the observations via the likelihood function. This is easily seen by moving from the top to the bottom of the left column of Figure \ref{fig:figure1}, corresponding to sample sizes of, respectively, $n=50$, $n=100$ and $n=1000$.

The right column of Figure \ref{fig:figure1} shows the frequentist coverage of the $95\%$ credible intervals for $n=50$ (top), $n=100$ (middle) and $n=1000$ (bottom). Although there is no apparent impact from different values of $\alpha$, we note more variable coverage values around the region of the parameter space associated with higher relative mean squared errors, compared to areas where the relative mean squared error is smaller. Although this behaviour is common to any sample size, we see a shifting of the more uncertain area towards higher number of degrees of freedom. This last result is common to other Bayesian objective methods to estimate $\nu$, such as in \cite{Fonseca:Ferreira:Migon:2008}, \cite{Villa:Walker:2014a} and the references therein. The above conclusion can also be drawn by inspecting the posterior median credible intervals for $\nu$, shown in Table \ref{tab:median_intervals}. In particular, the higher the value of the number of degrees of freedom, keeping the sample size fixed, the larger the interval. Alternatively, the higher the sample size, for the same value of $\nu$, the smaller the median credible interval. As expected, there is no appreciable difference in the median interval for different values of $\alpha$.\\

\begin{table}[]
\centering
\begin{tabular}{c|ccc|ccc|ccc}
\hline
      & \multicolumn{3}{c|}{$n=50$}                 & \multicolumn{3}{c|}{$n=100$}                & \multicolumn{3}{c}{$n=1000$}               \\
      \hline
$\nu$ & $\alpha=0.3$ & $\alpha=0.5$ & $\alpha=0.8$ & $\alpha=0.3$ & $\alpha=0.5$ & $\alpha=0.8$ & $\alpha=0.3$ & $\alpha=0.5$ & $\alpha=0.8$ \\
\hline
1     & (1,1)        & (1,1)        & (1,1)        & (1,1)        & (1,1)        & (1,1)        & (1,1)        & (1,1)        & (1,1)        \\
2     & (2,4)        & (1,5)        & (2,3)        & (1,2)        & (2,3)        & (2,3)        & (2,2)        & (2,2)        & (2,2)        \\
3     & (2,6)        & (2,8)        & (2,7)       & (2,3)        & (2,4)        & (2,6)        & (3,3)        & (3,3)        & (3,3)        \\
4     & (2,7)        & (2,7)       & (2,6)        & (3,9)        & (2,5)        & (3,7)        & (4,5)        & (4,5)        & (4,5)        \\
5     & (4,8)       & (3,8)       & (2,8)        & (4,14)       & (3,10)       & (4,14)       & (4,5)        & (4,5)        & (4,5)        \\
6     & (2,25)        & (3,7)       & (3,20)       & (4,15)       & (4,12)       & (4,17)       & (5,6)        & (5,6)        & (5,7)        \\
7     & (3,29)       & (5,29)       & (5,30)       & (3,10)       & (4,15)       & (4,14)       & (5,7)        & (6,8)        & (5,8)        \\
8     & (2,26)       & (3,29)       & (3,26)       & (3,28)        & (5,27)       & (5,29)       & (6,10)       & (6,9)        & (6,10)       \\
9     & (5,30)       & (3,30)       & (5,30)       & (3,22)       & (4,24)       & (4,23)       & (7,11)       & (7,11)       & (7,12)       \\
10    & (3,26)       & (4,30)       & (4,29)       & (7,30)       & (5,28)       & (5,28)       & (7,12)       & (9,16)       & (7,12)       \\
11    & (5,30)       & (5,30)       & (3,27)       & (5,29)       & (5,28)       & (4,27)       & (9,18)       & (8,13)       & (8,14)       \\
12    & (4,27)       & (6,29)       & (6,30)       & (8,30)       & (6,29)       & (4,24)       & (8,13)       & (9,17)       & (9,16)       \\
13    & (5,30)       & (3,30)       & (6,30)       & (7,30)       & (5,28)       & (7,30)       & (9,15)       & (10,18)      & (10,20)      \\
14    & (4,30)       & (8,27)       & (6,30)       & (5,27)       & (6,29)       & (6,30)       & (12,29)      & (9,19)       & (9,19)       \\
15    & (6,30)       & (3,29)       & (4,29)       & (5,29)       & (6,30)       & (7,30)       & (9,18)       & (9,18)       & (9,17)       \\
16    & (5,30)       & (3,30)       & (6,30)       & (5,28)       & (8,30)       & (7,30)       & (13,29)      & (11,25)      & (11,22)      \\
17    & (4,29)       & (7,29)       & (5,30)       & (6,30)       & (5,27)       & (5,26)       & (9,19)       & (10,21)      & (11,24)      \\
18    & (5,30)       & (4,30)       & (7,30)       & (7,30)       & (6,29)       & (7,30)       & (13,30)      & (11,24)      & (13,29)      \\
19    & (5,30)       & (4,30)       & (4,30)       & (7,30)       & (6,30)       & (6,29)       & (11,23)      & (12,29)      & (13,29)      \\
20    & (4,29)       & (2,29)       & (5,29)       & (7,30)       & (6,30)       & (7,30)       & (15,30)      & (13,29)      & (13,29) \\
\hline     
\end{tabular}
\caption{Median $95\%$ credible intervals of the posterior of $\nu$, for simulations with $\nu=1,\ldots,20$, $\alpha=0.3,0.5,0.8$ and $n=50,100,1000$.}
\label{tab:median_intervals}
\end{table}

To have a feeling for the overall estimation procedure, we illustrate in detail the analysis of a single independent and identically distributed sample from a known model. We draw a sample of size $n=200$ from an AST distribution with parameters $\alpha=0.35$, $\mu=2$, $\sigma=1.5$ and $\nu=6$. The posterior distributions are obtained via Monte Carlo methods with 100000 iterations and a burn-in period of 5000 iterations and by considering the objective priors described in Section \ref{sc_priors}. In Figure \ref{fig:single_samp} we have reported, for each parameter, the chain samples and the histogram of the posterior distribution,
\begin{figure}[h!]
\centering
\subfigure[]{%
\includegraphics[scale=0.25]{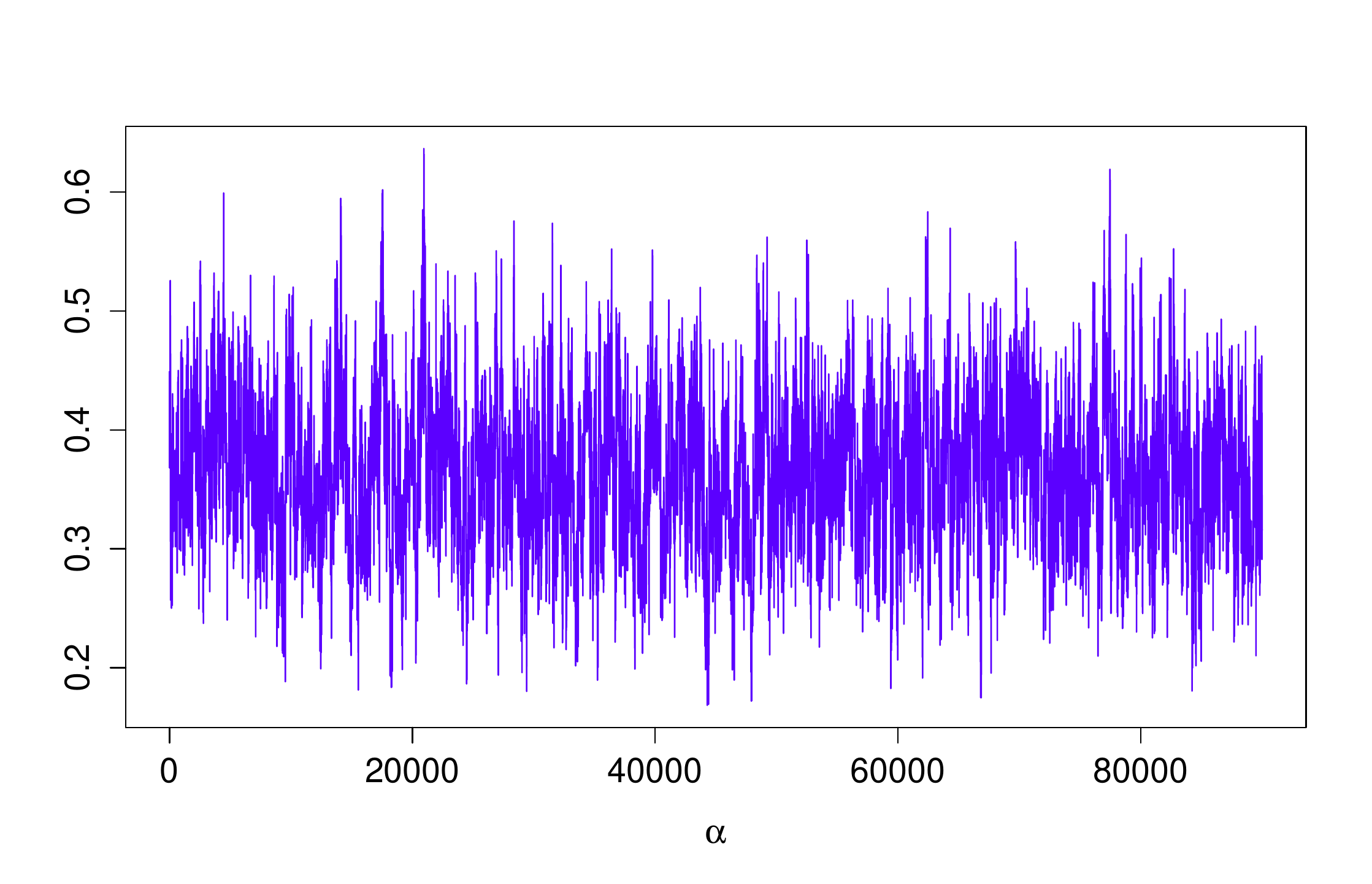} 
\label{fig:sf1}}
\quad
\subfigure[]{%
\includegraphics[scale=0.25]{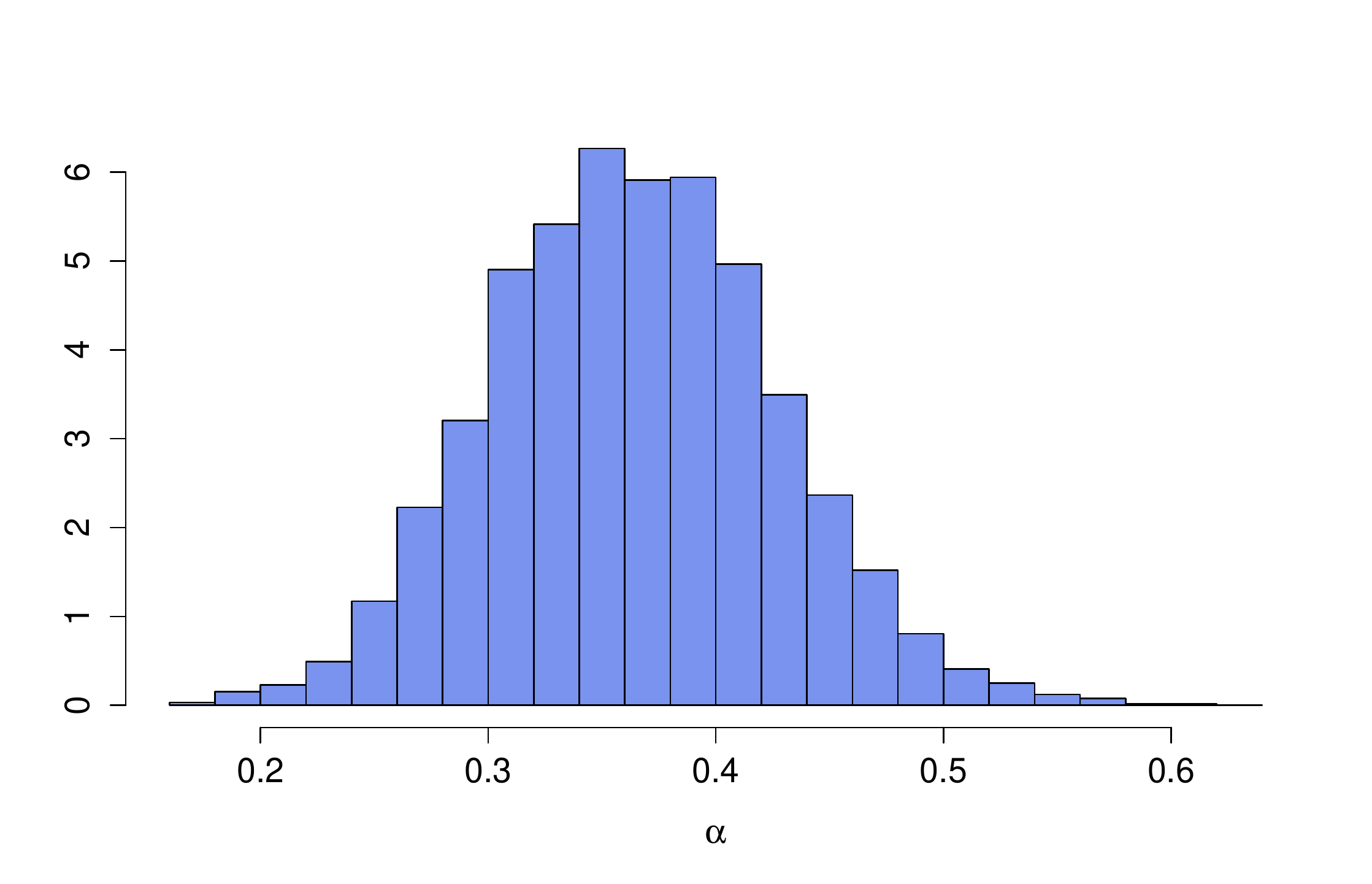} 
\label{fig:sf2}}
\subfigure[]{%
\includegraphics[scale=0.25]{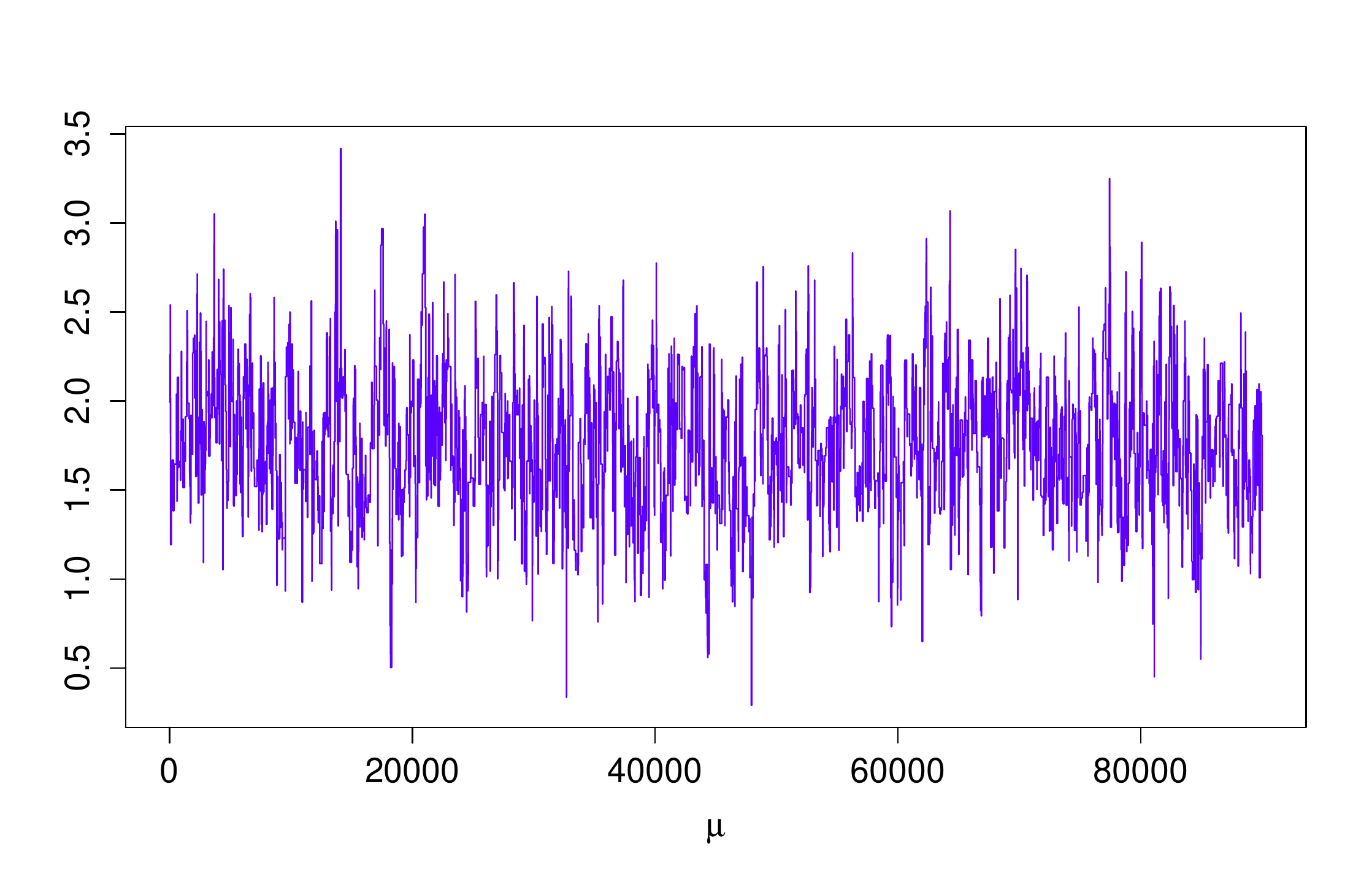}
\label{fig:sf3}}
\quad
\subfigure[]{%
\includegraphics[scale=0.25]{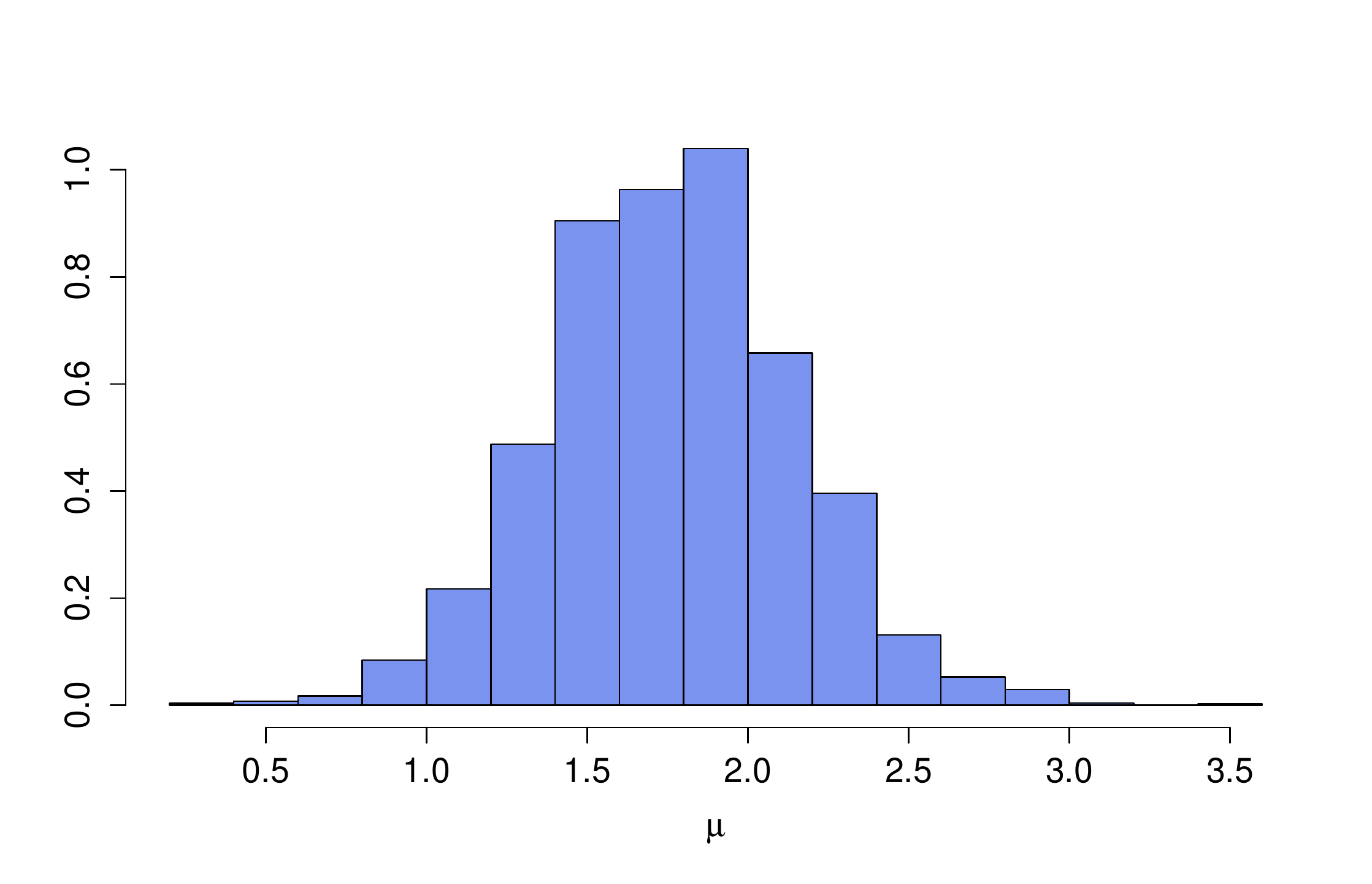}
\label{fig:sf4}}
\subfigure[]{%
\includegraphics[scale=0.25]{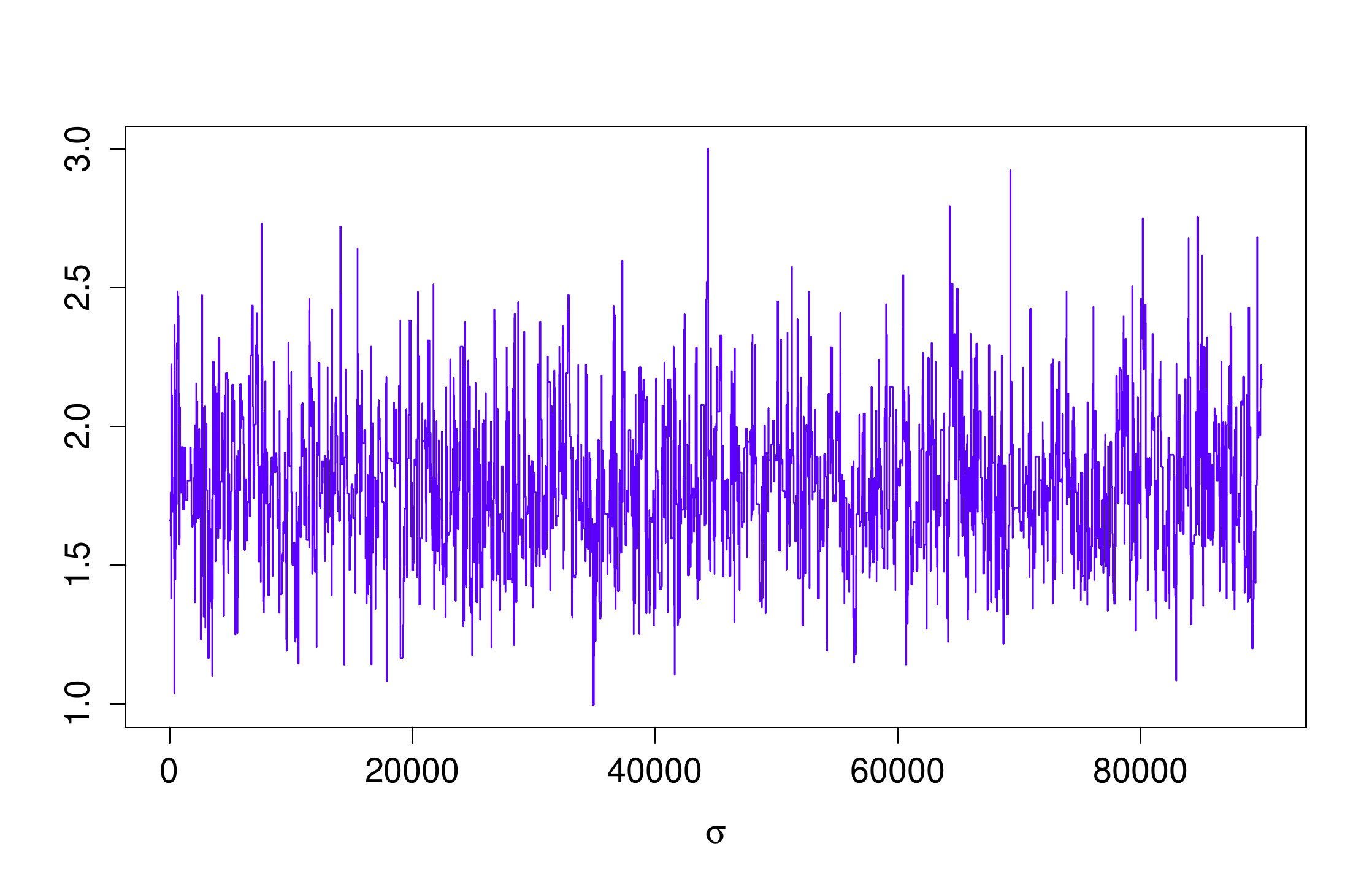}
\label{fig:sf5}}
\quad
\subfigure[]{%
\includegraphics[scale=0.25]{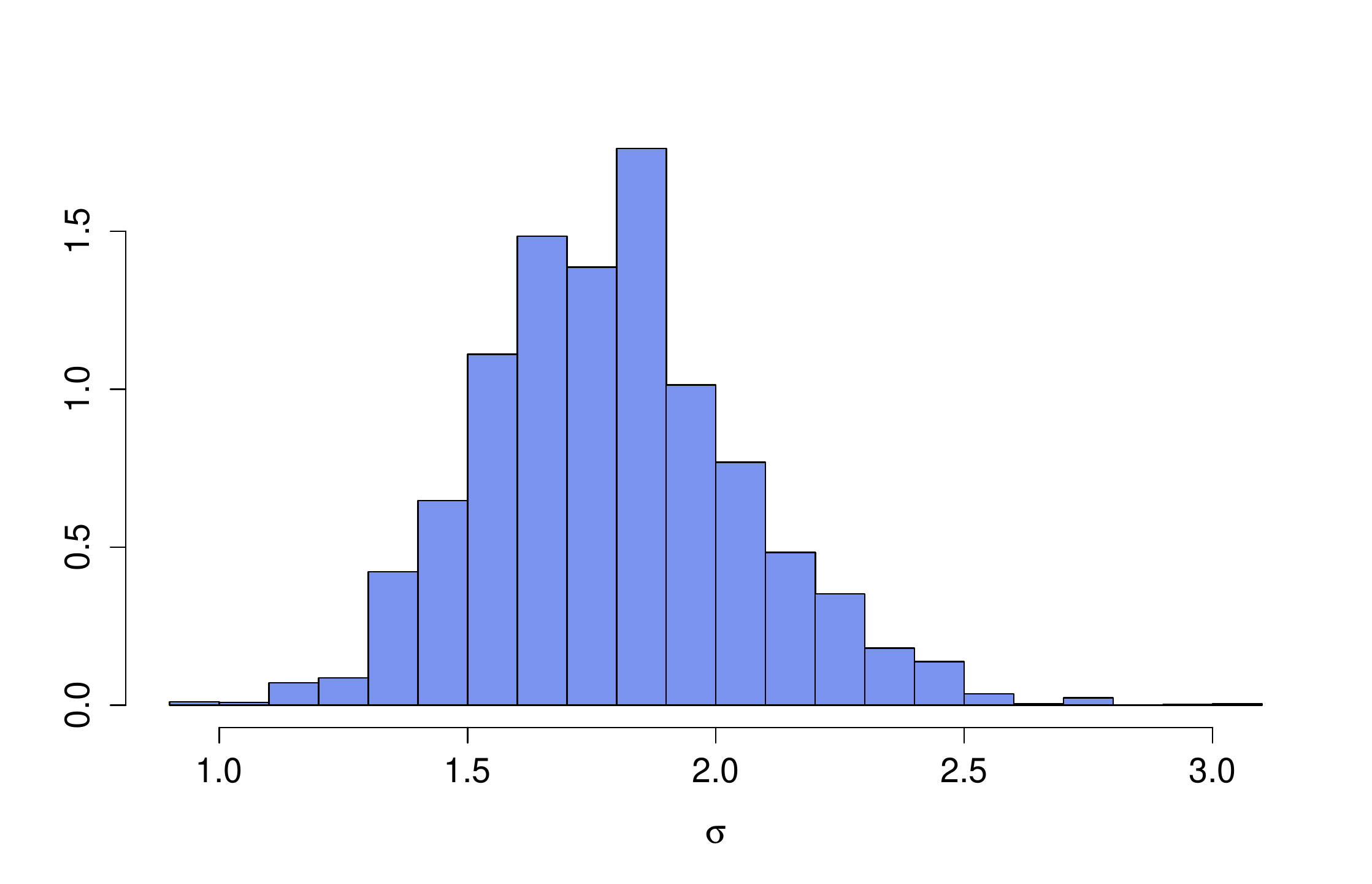}
\label{fig:sf6}}
\subfigure[]{%
\includegraphics[scale=0.25]{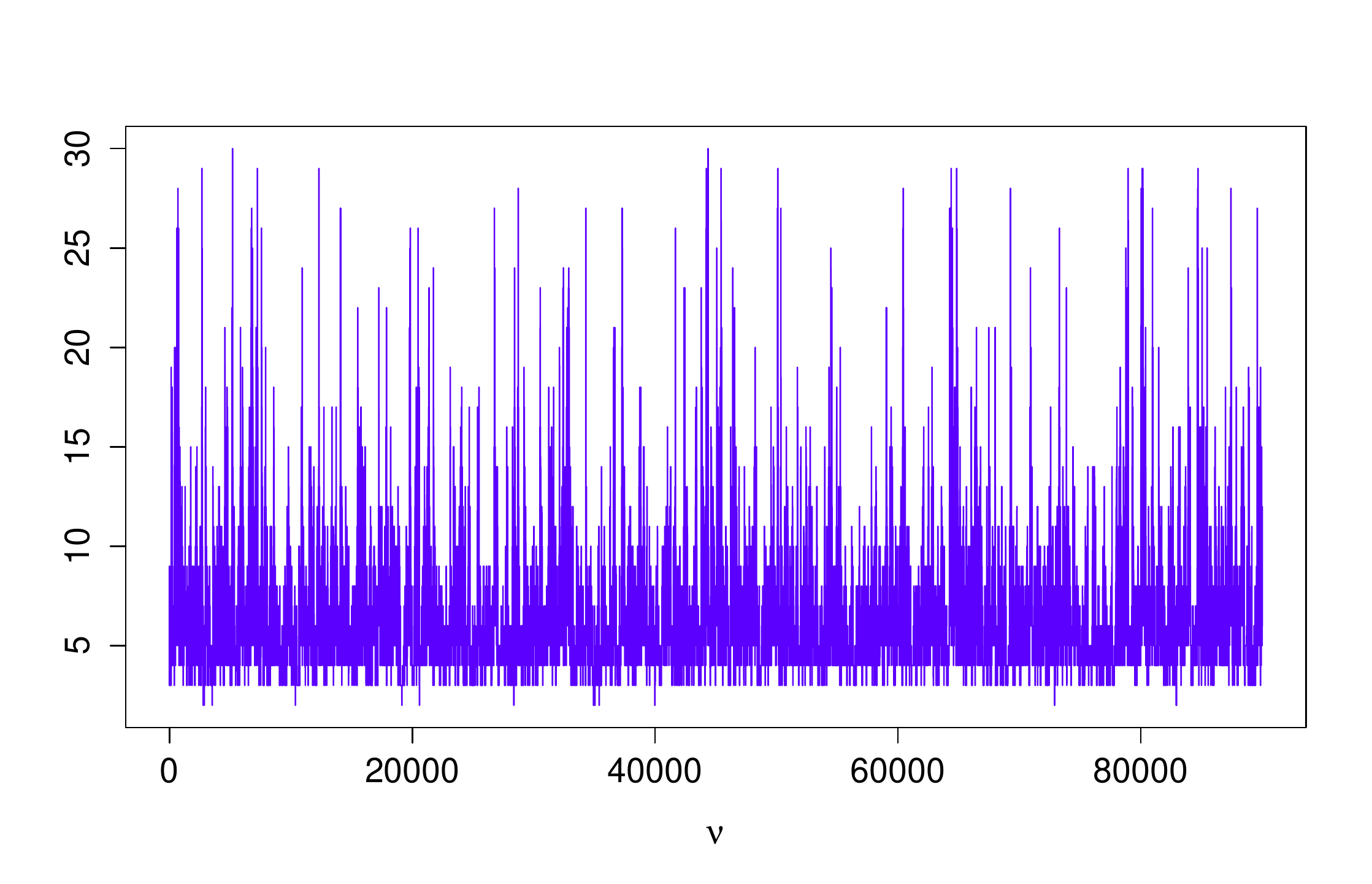}
\label{fig:sf7}}
\quad
\subfigure[]{%
\includegraphics[scale=0.25]{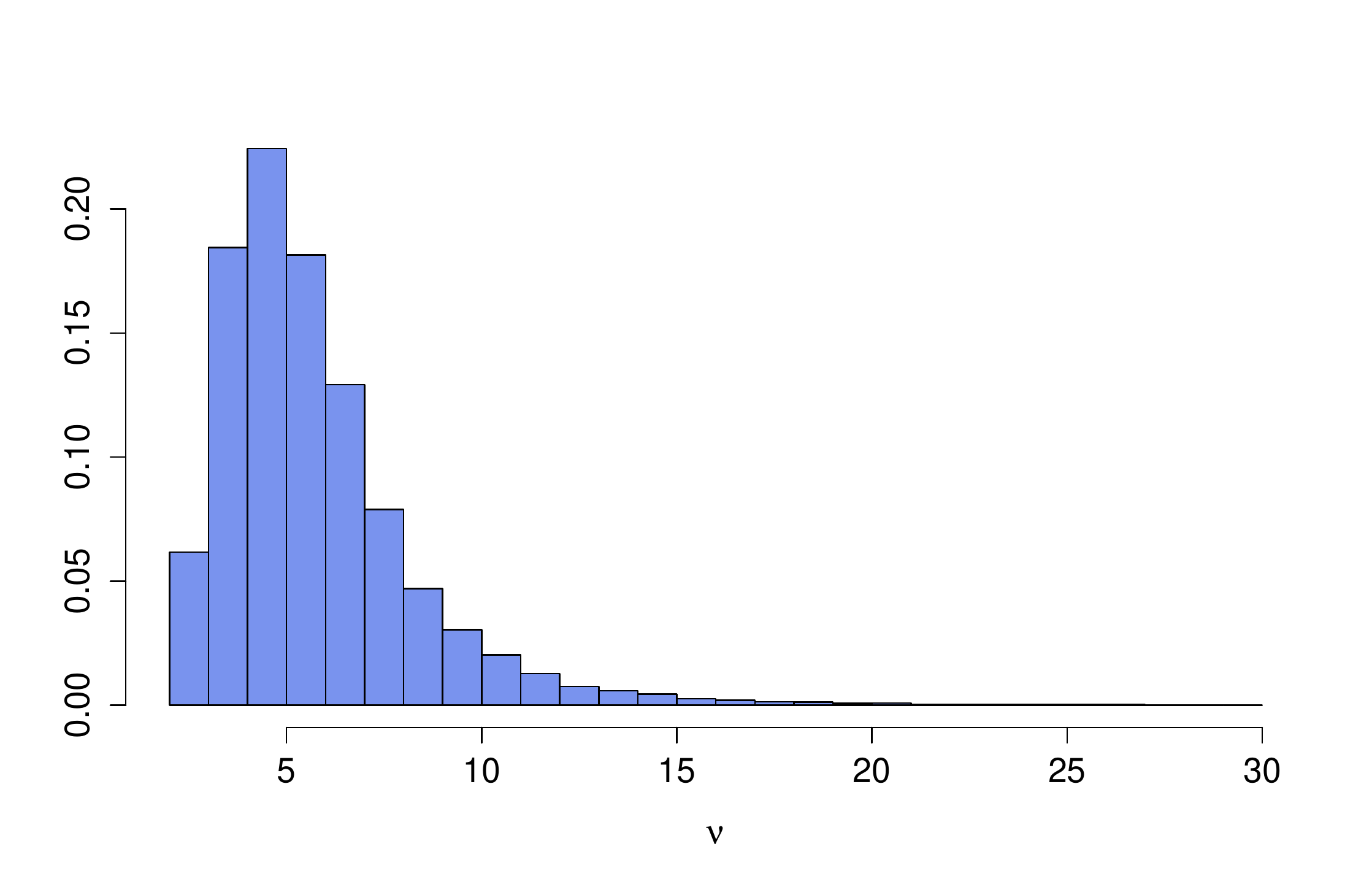}
\label{fig:sf8}}
\caption{Sample chains (left graphs) and histograms of the posterior distributions (right graphs) of the parameters for the simulated data from the AST with $\alpha=0.35$, $\mu=2$, $\sigma=1.5$ and $\nu=6$.}
\label{fig:single_samp}
\end{figure}
while in Table \ref{tab:singlesample} we have the summary statistics of each posterior. In particular, we have computed the posterior mean, the posterior median and the $95\%$ credible interval of the posterior distribution. 

\begin{table}
\centering
\begin{tabular}{cccc}
Parameter & Mean & Median & $95\%$ C.I. \\ 
\hline 
$\alpha$ & 0.36 & 0.36 & (0.25,0.49) \\ 
$\mu$ & 1.76 & 1.76 & (1.02,2.53) \\ 
$\sigma$ & 1.79 & 1.78 & (1.33,2.37) \\ 
$\nu$ & 6.24 & 6 & (3,13) \\ 
\hline 
\end{tabular}
\caption{Summary statistics of the posterior distributions for the parameters of the simulated data from an AST with $\alpha=0.35$, $\mu=2$, $\sigma=1.5$ and $\nu=6$.}
\label{tab:singlesample}
\end{table}
By inspecting the histograms and the summary statistics of the posterior distributions we can assess on the appropriateness of the inferential process. In particular, the mean (or the median for $\nu$) of the posteriors are very close to the true parameter values, which are well within the limits of the corresponding credible intervals.

\section{Real data analysis}\label{sc_real}
To show how the discussed model works in practice, we have chosen two well known data sets, both related to insurance loss. The first data set contains 2,167 individual losses each with a value of one million Danish Krone (DKK) or above, collected from January 1980 to December 1990 \citep{Mcneil:1997}. The second data set relates to 1,500 indemnity payments, in thousand of US dollars \citep{FreVal:1998}.

\subsection{Danish fire losses}\label{sc_danish}
The Danish fire loss data set contains losses due to fire with a single value of at least DKK 1 million. Table \ref{tab:danish_loss1} reports some descriptive statistics of the data, both in the nominal scale and in the log-scale. It is common practice, when analysing insurance data (and not only), to consider the logarithm of the data for modelling purposes as this results in a reduction of skewness \citep{Bolance:2008}. Although the skewness index is drastically reduced by the log-transform of the Danish loss data, as shown in Table \ref{tab:danish_loss1}, its value still indicates a significantly positive skewness. The above result is easily noticeable by inspecting the histograms of the data in Figure \ref{fig:Danish_hist}.
\begin{table}
\centering
\begin{tabular}{lcc}
 & Danish & Danish (log-scale) \\ 
\hline 
Mean & 3.39 & 0.79 \\ 
Standard Deviation & 8.51 & 0.72 \\ 
Skewness & 18.75 & 1.76 \\ 
\hline 
\end{tabular} 
\caption{Descriptive statistics of the Danish loss data set in millions of Danish Krone (left) and in the log-scale (right).}
\label{tab:danish_loss1}
\end{table}
Suitable statistical tests can be performed to support the conclusion of departure for normality, such as the Jarque--Bera test for normality \citep{JarBer:1980}, the D'agostino test for skewness \citep{Dago:1970} or the Anscombe-Glynn test of kurtosis \citep{AnsGly:1983}, for example (results not reported here).
\begin{figure}[h!]
\centering
\subfigure[]{%
\includegraphics[scale=0.30]{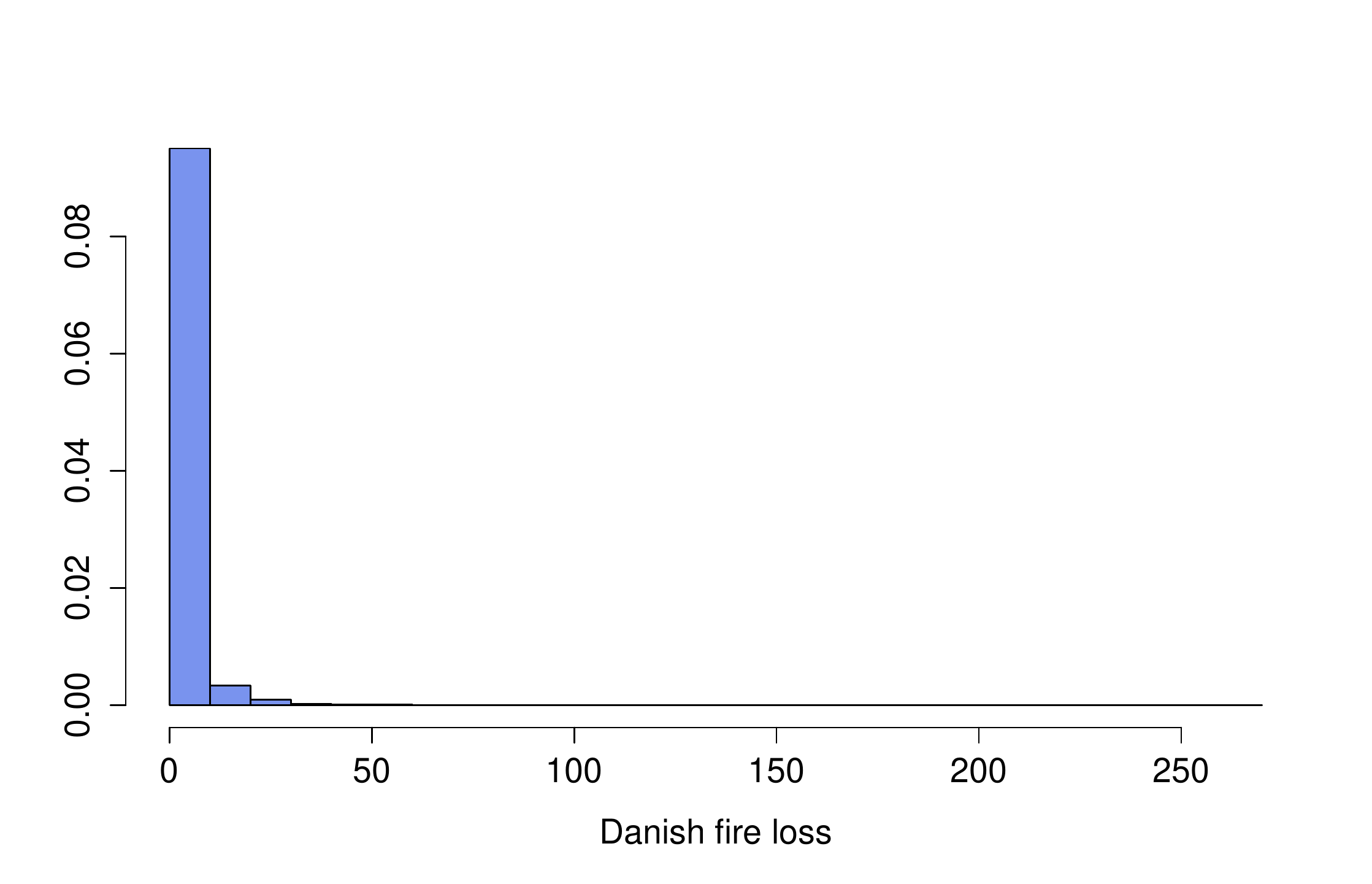}  
\label{fig:sub1}}
\quad
\subfigure[]{%
\includegraphics[scale=0.30]{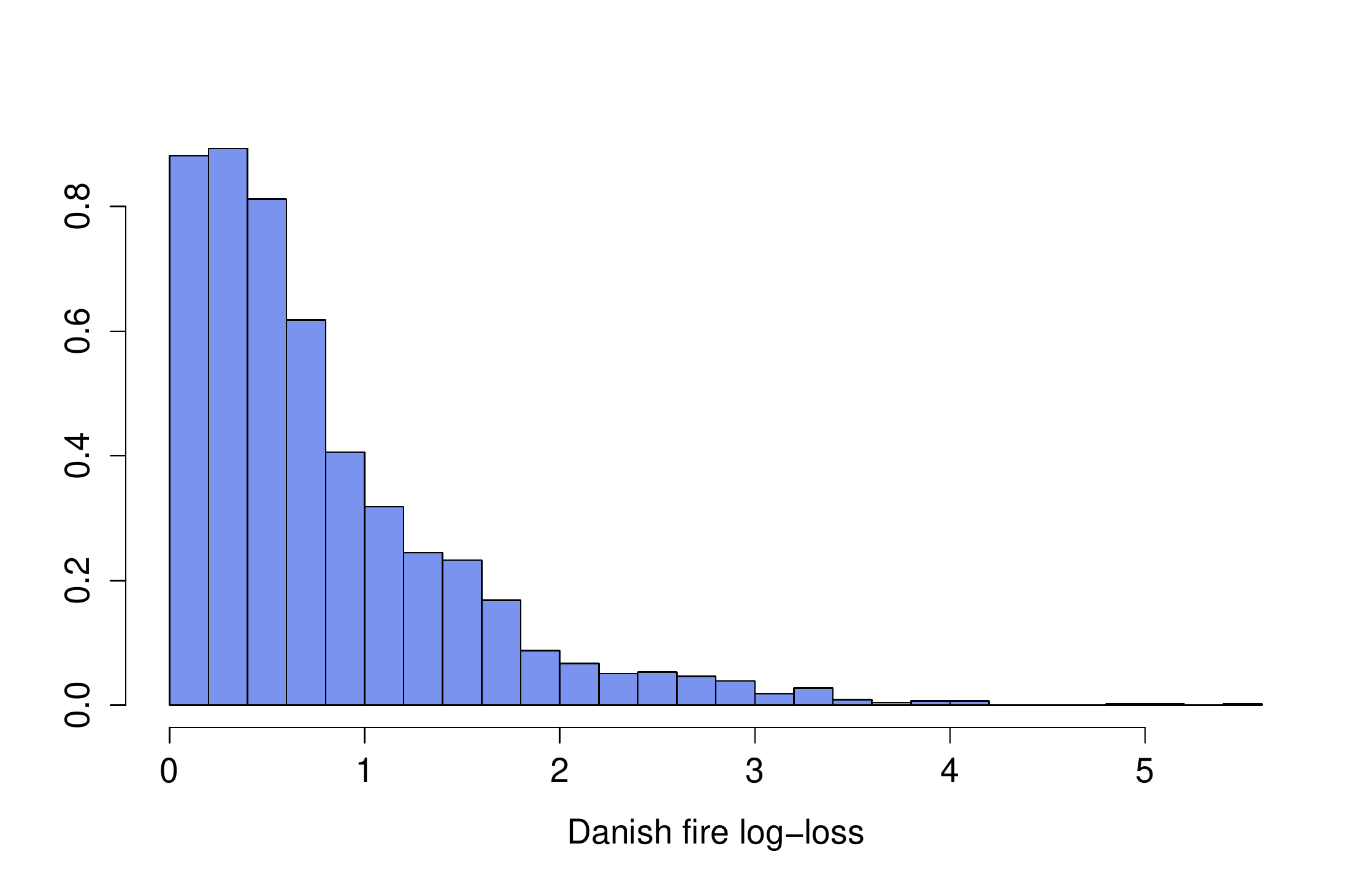} 
\label{fig:sub2}}
\caption{Histograms of the Danish fire loss data (left) and of the same data set in the log-scale (right).}
\label{fig:Danish_hist}
\end{figure}

The prior distributions used to analyse the Danish fire loss data set where in line with the overall objective approach discussed in the paper. In particular, we used the Jeffreys' prior for the skewness parameter $\alpha$, that is $\pi(\alpha)\sim \mbox{Beta}(1/2,1/2)$, the discrete truncated prior for the number of degrees of freedom $\nu$, and the reference prior for the pair location-scale parameters $(\mu,\sigma)$, that is $\pi(\mu,\sigma)\propto 1/\sigma$.

\begin{figure}[h!]
\centering
\includegraphics[scale=0.30]{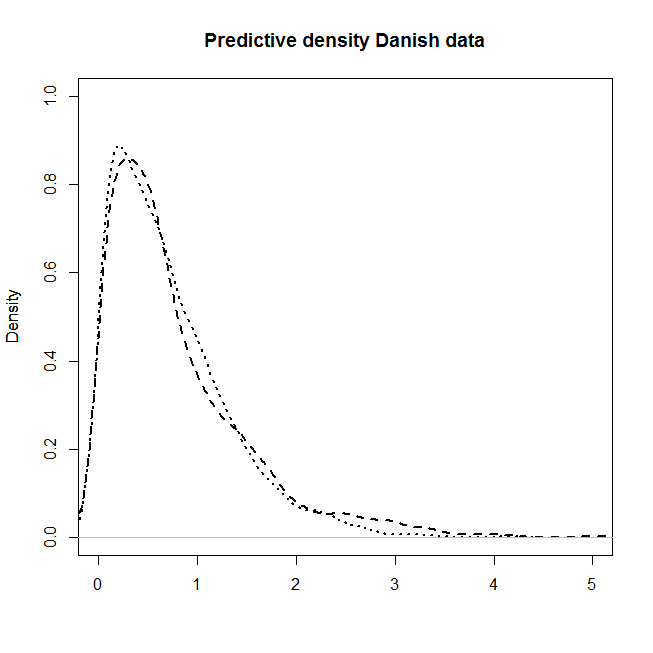}  
\caption{Real Data (dashed line) vs Posterior Predictive Distribution (dotted line) of the Danish Loss data. \label{fig:predDan}}
\end{figure}

The marginal posterior distributions for the parameters, as they are analytically intractable, have been obtained via Monte Carlo methods by applying the algorithm described in the Appendix. We have run multiple chains for each parameter, with different sparse starting points. In particular, we have run 
500000 iterations and considered a burn in of 100000. The convergence has been assessed by computing the Gelman and Rubin's statistics \citep{BroGel:1998,GelRub:1992} and monitoring the posterior running means.
\begin{figure}[h]
\centering
\subfigure[]{%
\includegraphics[scale=0.30]{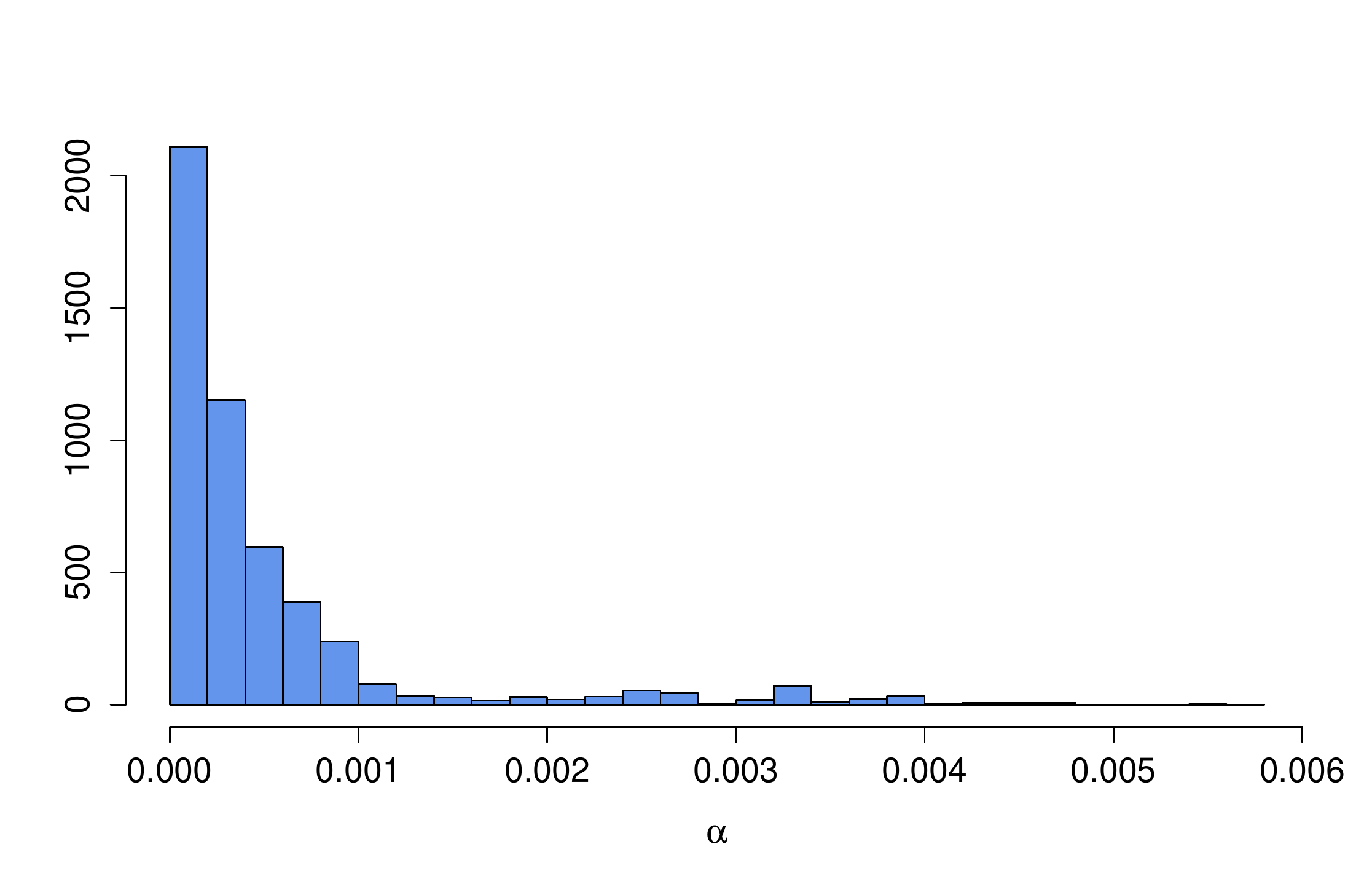} 
\label{fig:dan_hist_sub1}}
\quad
\subfigure[]{%
\includegraphics[scale=0.30]{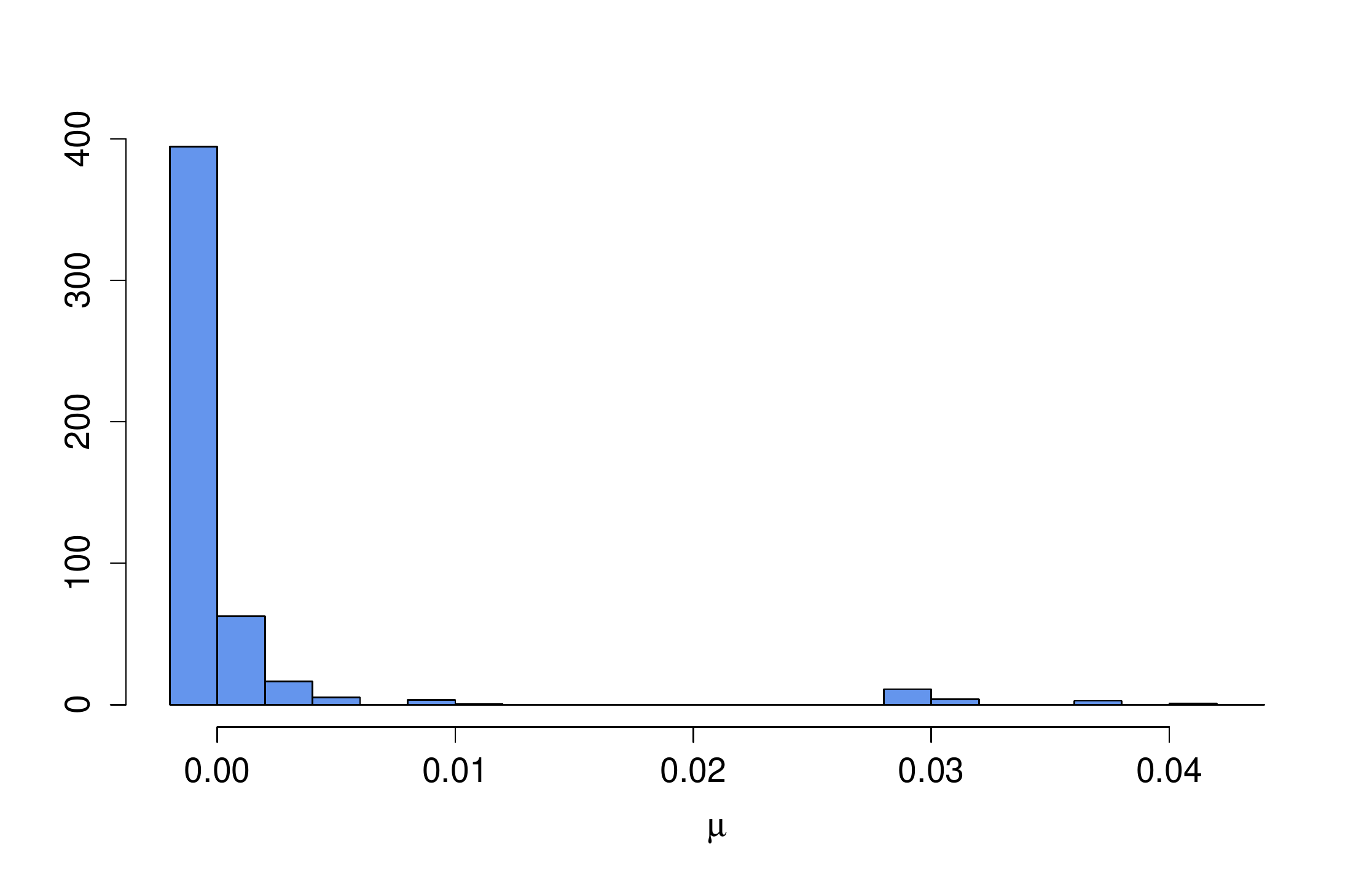} 
\label{fig:dan_hist_sub2}}
\subfigure[]{%
\includegraphics[scale=0.30]{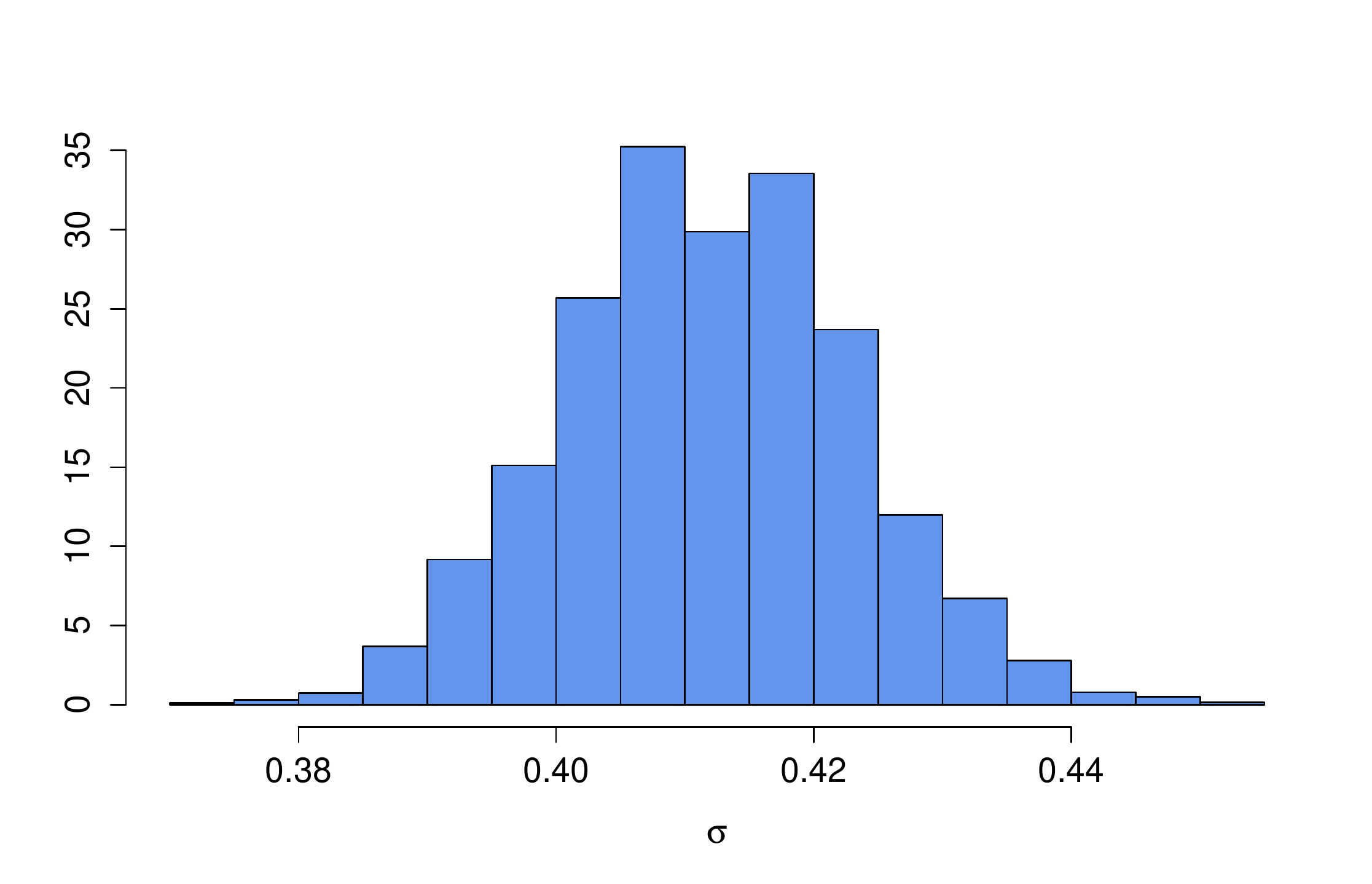}
\label{fig:dan_hist_sub3}}
\quad
\subfigure[]{%
\includegraphics[scale=0.30]{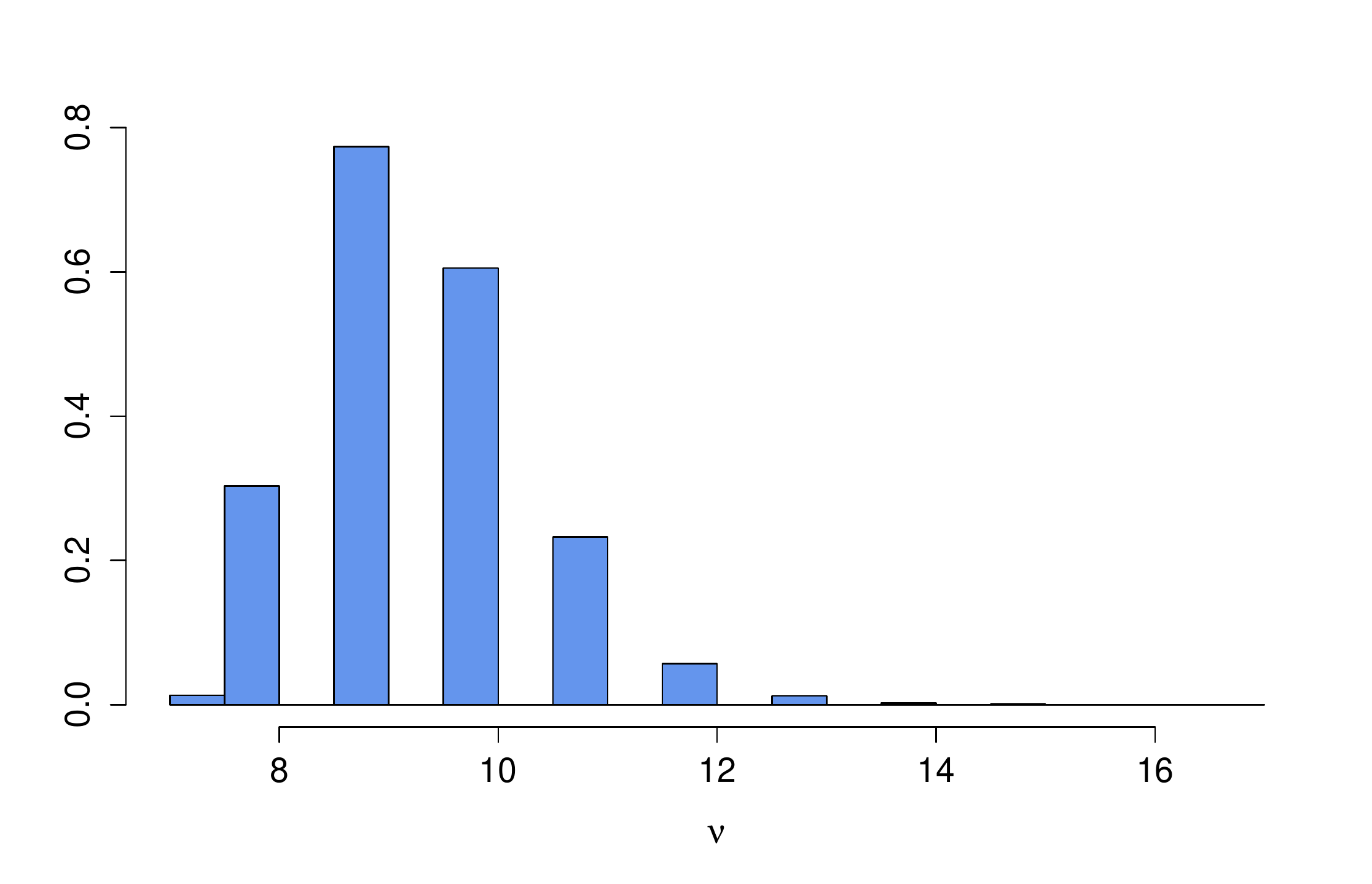}
\label{fig:dan_hist_sub4}}
\caption{Histograms of the posterior distributions for $\alpha$, $\mu$, $\sigma$ and $\nu$ for the Danish fire loss data set (in the log-scale).}
\label{fig:danish_post_hist}
\end{figure}
We have reported the histograms of the posterior marginal distributions for the parameters of the model in Figure \ref{fig:danish_post_hist}, and the corresponding summary statistics in Table \ref{tab:danish_post_summary}.
\begin{table}
\centering
\begin{tabular}{cccc}
\hline 
Parameter & Mean & Median & $95\%$ C.I. \\ 
\hline 
$\alpha$ & 0.0005 & 0.0003 & (0.0001,0.0020) \\ 
$\mu$ & 0.0042 & 0.0004 & (0.0000,0.0378) \\ 
$\sigma$ & 0.4142 & 0.4123 & (0.3906,0.4359) \\ 
$\nu$ & 9.4900 & 9 & (8,12) \\ 
\hline 
\end{tabular} 
\caption{Summary statistics of the posterior distribution of the Danish fire loss data set in the log-scale.}
\label{tab:danish_post_summary}
\end{table}
As expected, the value of the skewness parameter is very close to zero. In fact, both from the data histogram and summary statistics, it is possible to deduce that the data has a strong positive skewness. The median of the posterior of the number of degrees of freedom $\nu$ indicates a heavy-tailed behaviour in the observations. Both the skewness and heavy-tail results are consistent with the expected behaviour of insurance loss data, even after the data has been log-transformed, in this case. Figure \ref{fig:predDan} shows the posterior predictive density against the real data. We note no substantial difference in the two curves. Finally, to compare the data statistics of Table \ref{tab:danish_loss1} with the MCMC estimation, we have computed the Monte Carlo estimates (in the log-scale) of the mean, the standard deviation and the skewness index, obtaining, respectively, the values 0.79, 0.72, 1.77.

\subsection{US indemnity loss}\label{sc_US}
The second data set we analyse is widely used in the literature and it contains US indemnity losses publicly available \citep{FreVal:1998}. It contains 1,500 liability claims each of which with an associated indemnity payment in thousands of US dollars (USD).

\begin{table}
\centering
\begin{tabular}{lcc}
 & US & US (log-scale) \\ 
\hline 
Mean & 41.21 & 2.46 \\ 
Standard Deviation & 102.75 & 1.64 \\ 
Skewness & 9.16 & -0.15 \\ 
\hline 
\end{tabular} 
\caption{Descriptive statistics of the US loss data set in thousands of USD (left) and in the log-scale (right).}
\label{tab:usloss1}
\end{table}

Table \ref{tab:usloss1} shows the descriptive statistics of the US loss data set, both in thousands of USD and in the log-scale. As we did for the Danish fire loss data, we perform the analysis on the log-transformation of the observed values. Figure \ref{fig:US_hist} shows the histogram of the US loss data (left plot) and of the same data in the log-scale (right plot). From the first histogram on the left it is possible to see the typical behaviour of this type of data, that is a relatively high number of losses with a small value and a few losses of large value.
\begin{figure}[ht]
\centering
\subfigure[]{%
\includegraphics[scale=0.30]{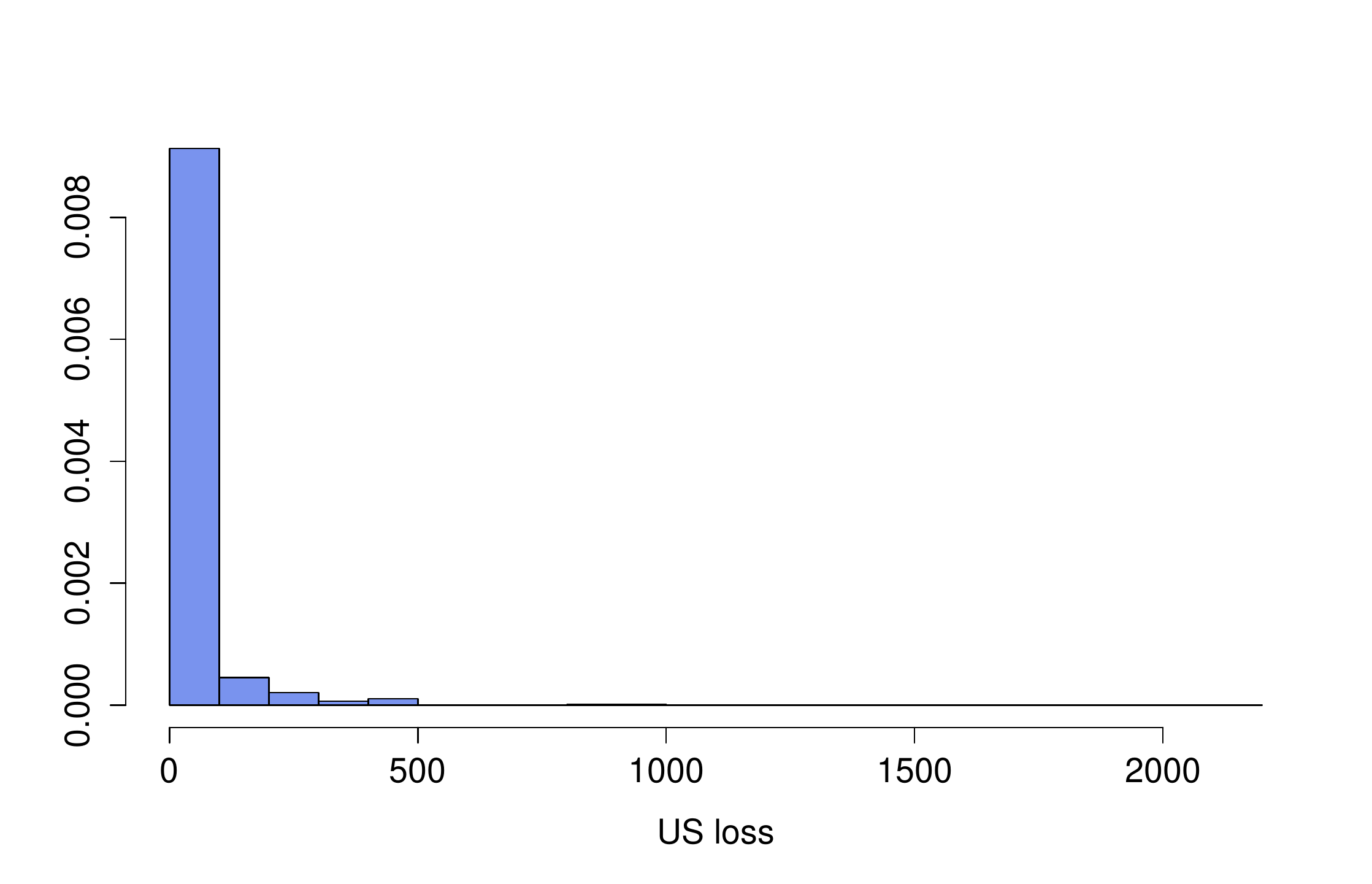}  
\label{fig:subfigur1}}
\quad
\subfigure[]{%
\includegraphics[scale=0.30]{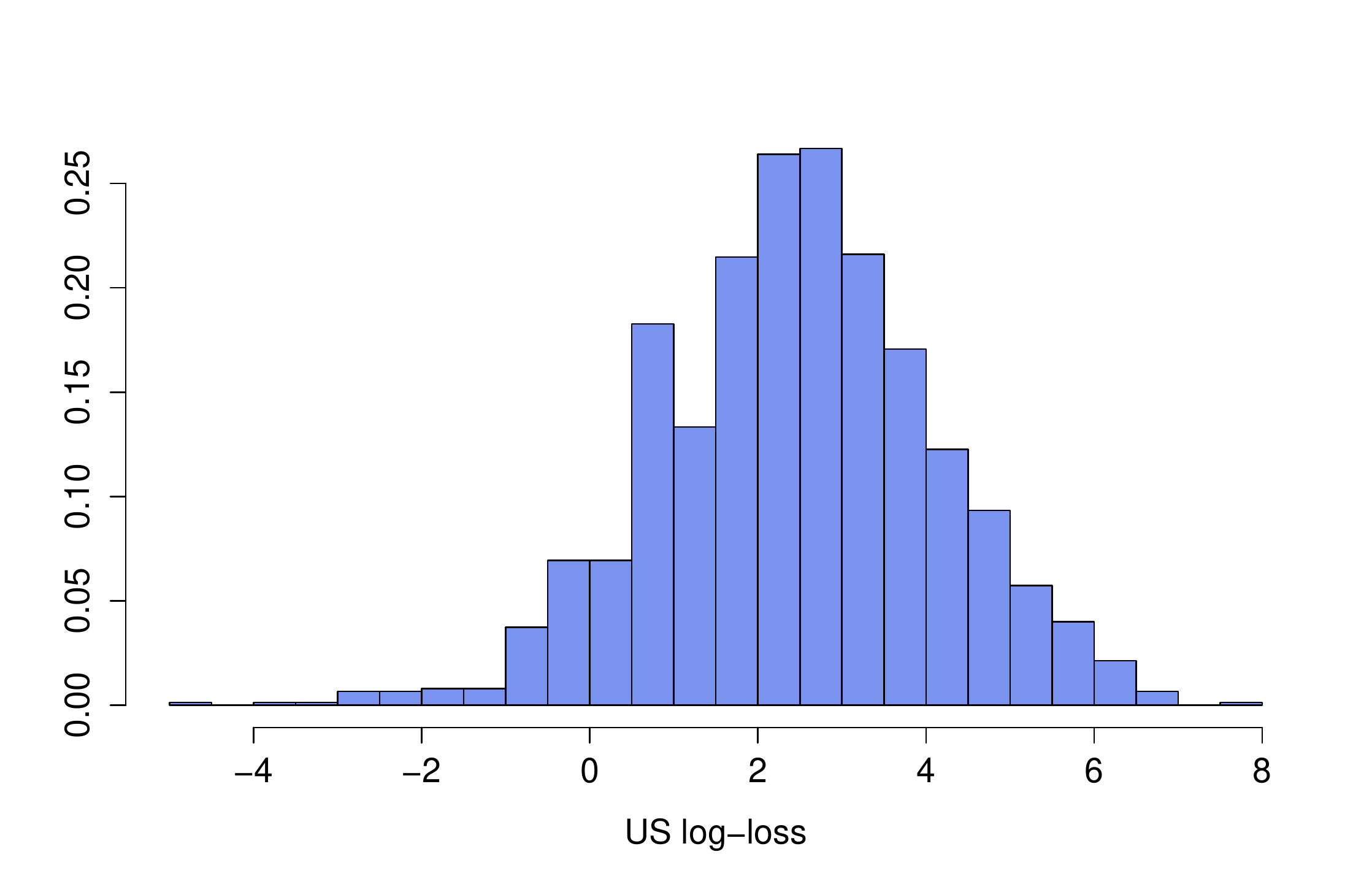} 
\label{fig:subfigur2}}
\caption{Histograms of the US loss data (left) and of the same data set in the log-scale (right).}
\label{fig:US_hist}
\end{figure}
The procedure followed to analyse the US loss data, in the log-scale, is analogous to the one employed to analyse the Danish fire loss data set in Section \ref{sc_danish}. The histograms of the posterior distributions of the parameters of the model are in Figure \ref{fig:US_post_hist}, with the corresponding summary statistics in Table \ref{tab:US_post_summary}.
\begin{figure}[h]
\centering
\subfigure[]{%
\includegraphics[scale=0.30]{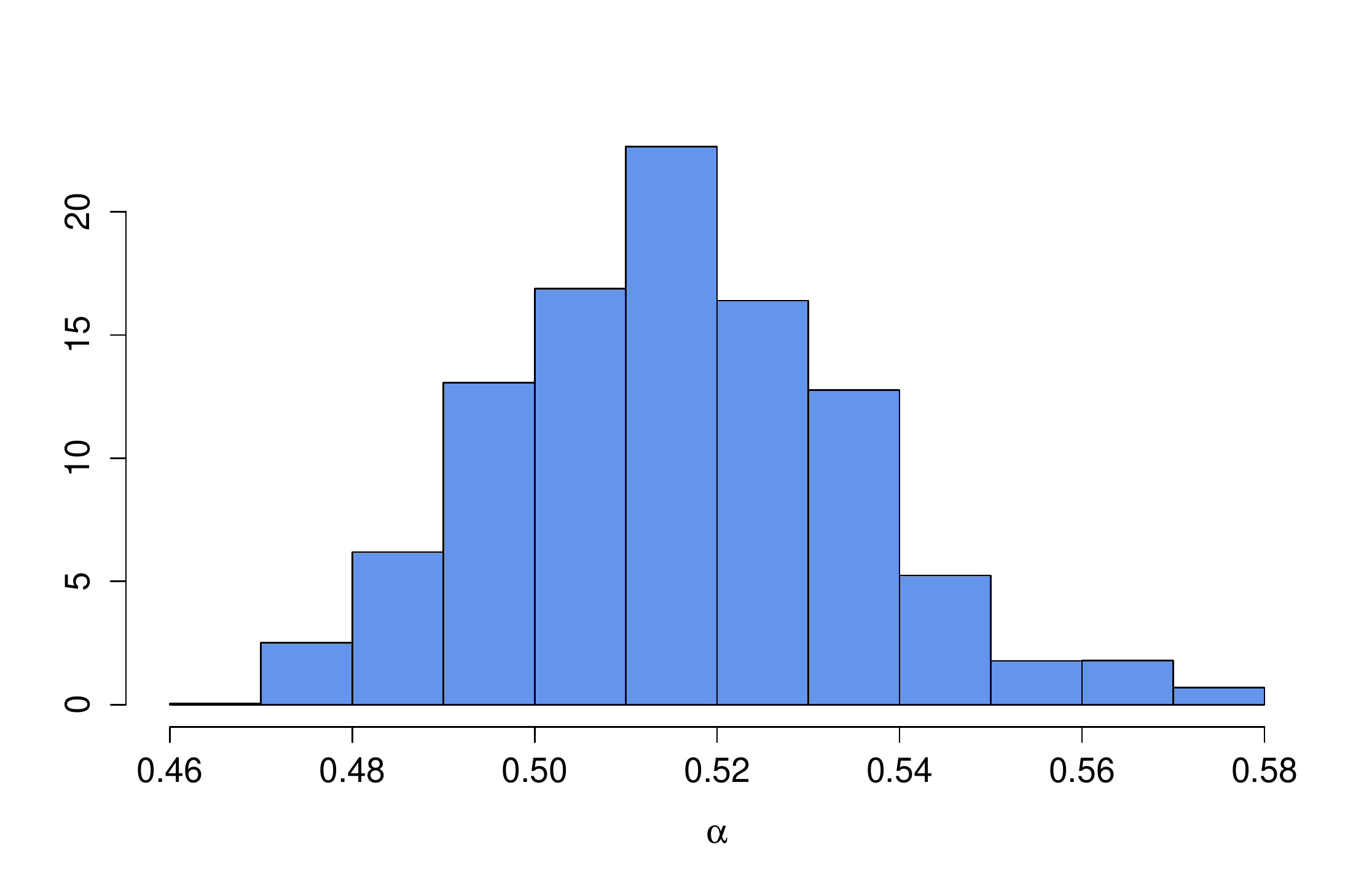} 
\label{fig:US_hist_sub1}}
\quad
\subfigure[]{%
\includegraphics[scale=0.30]{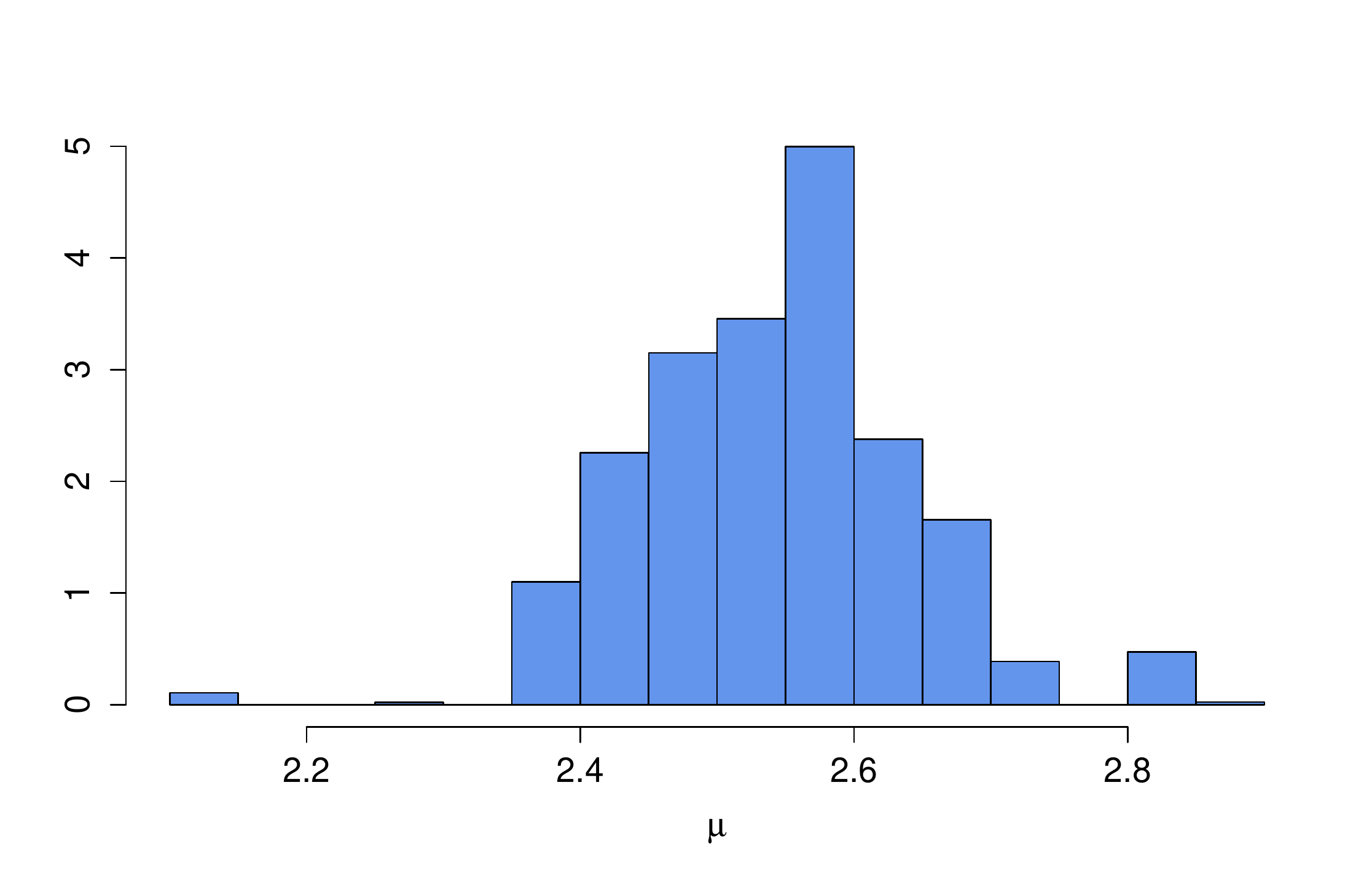} 
\label{fig:US_hist_sub2}}
\subfigure[]{%
\includegraphics[scale=0.30]{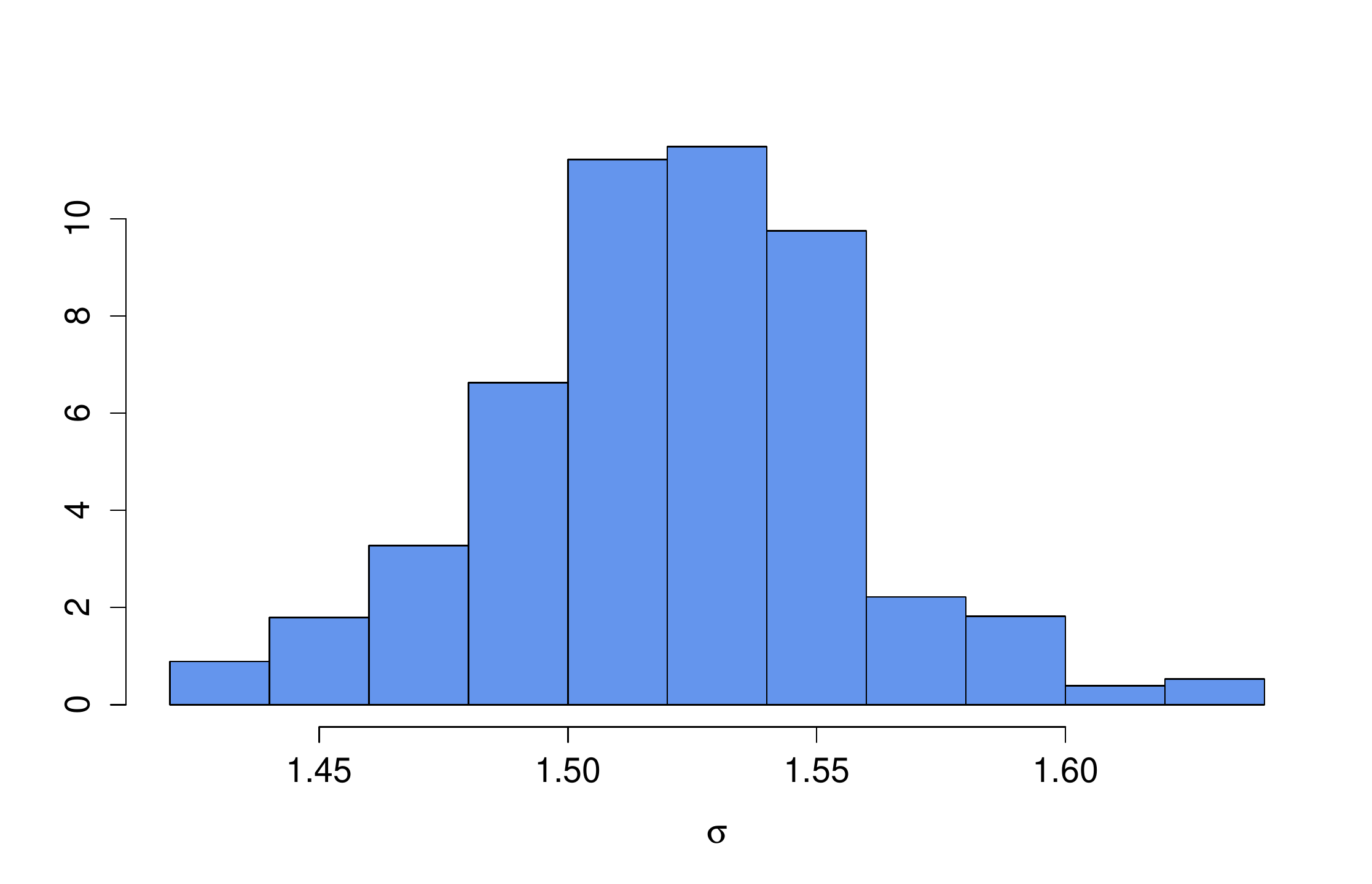}
\label{fig:US_hist_sub3}}
\quad
\subfigure[]{%
\includegraphics[scale=0.30]{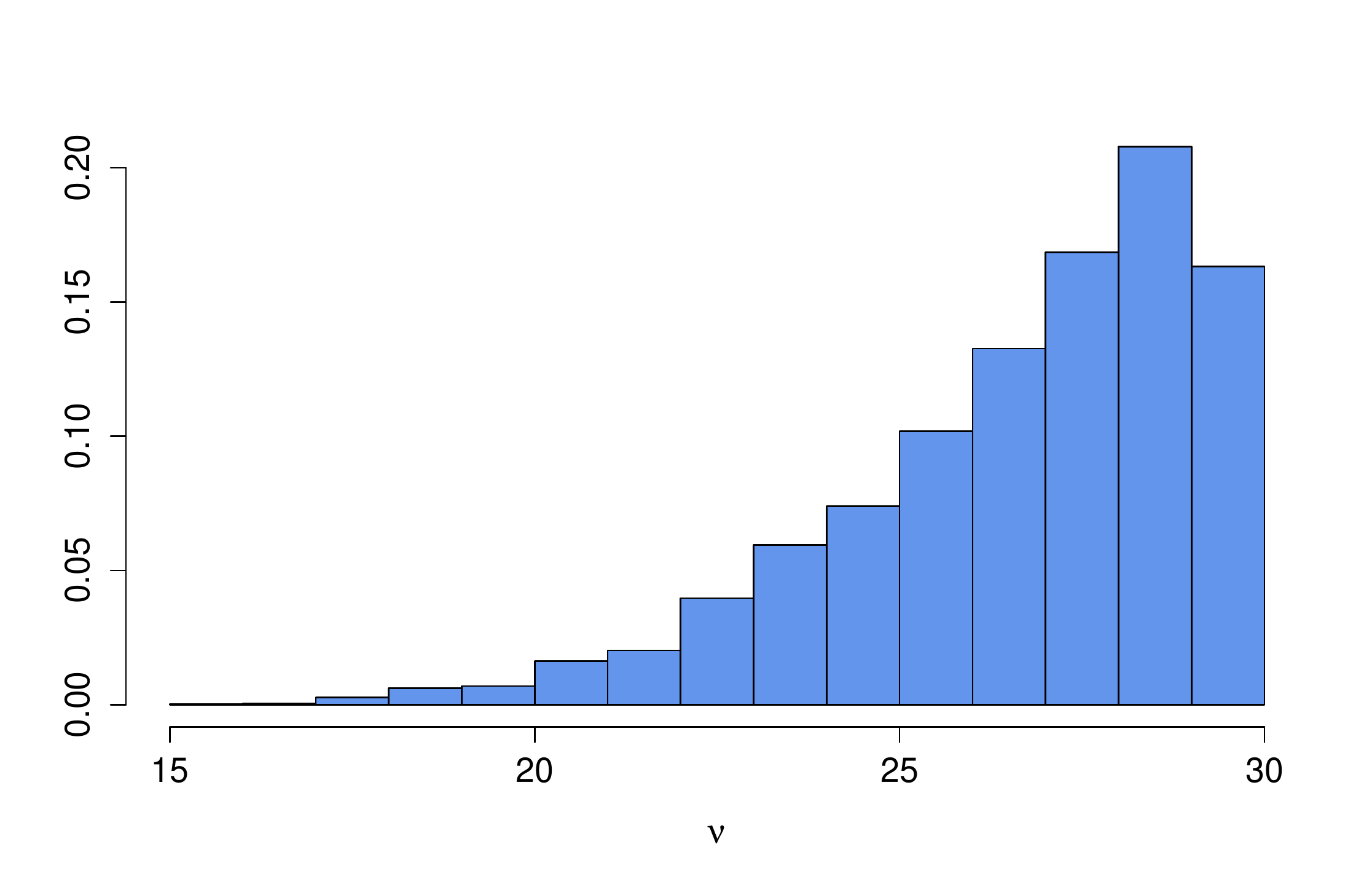}
\label{fig:US_hist_sub4}}
\caption{Histograms of the posterior distributions for $\alpha$, $\mu$, $\sigma$ and $\nu$ for the US loss data set (in the log-scale).}
\label{fig:US_post_hist}
\end{figure}
\begin{table}
\centering
\begin{tabular}{cccc}
\hline 
Parameters & Mean & Median & $95\%$ C.I. \\ 
\hline 
$\alpha$ & 0.52 & 0.52 & (0.48,0.56) \\ 
$\mu$ & 2.47 & 2.48 & (2.01,2.65) \\ 
$\sigma$ & 1.52 & 1.52 & (1.44,1.59) \\ 
$\nu$ & 27.16 & 28 & (21,30) \\ 
\hline 
\end{tabular} 
\caption{Summary statistics of the posterior distribution of the US loss data set (in the log-scale).}
\label{tab:US_post_summary}
\end{table}
We note the following two important results. First, the skewness parameter is estimated to be close to 0.5. This indicates that a symmetric model would be suitable for this data set. Second, the estimated number of degrees of freedom is very closed to the upper bound of the parameter space. This is a clear indication that the data could be represented by a $t$ density with a high number of degrees of freedom or, which is equivalent, by a normal density.

For what it concerns the inferential results, we see that the credible intervals are relatively narrow, indicating a relatively strong posterior beliefs about the obtained estimates. Figure \ref{fig:predUS} shows the posterior predictive density against the real data, where again we do not see any substantial difference between the two curves. As done in Section \ref{sc_danish}, we have performed Monte Carlo estimates of the data statistics fro the US Loss data set. In particular, we have computed the mean, the standard deviation and the skewness index obtaining, respectively, 2.45, 1.64, -0.15. These values can be compared with the ones in Table \ref{tab:usloss1}.

\begin{figure}[h!]
\centering
\includegraphics[scale=0.30]{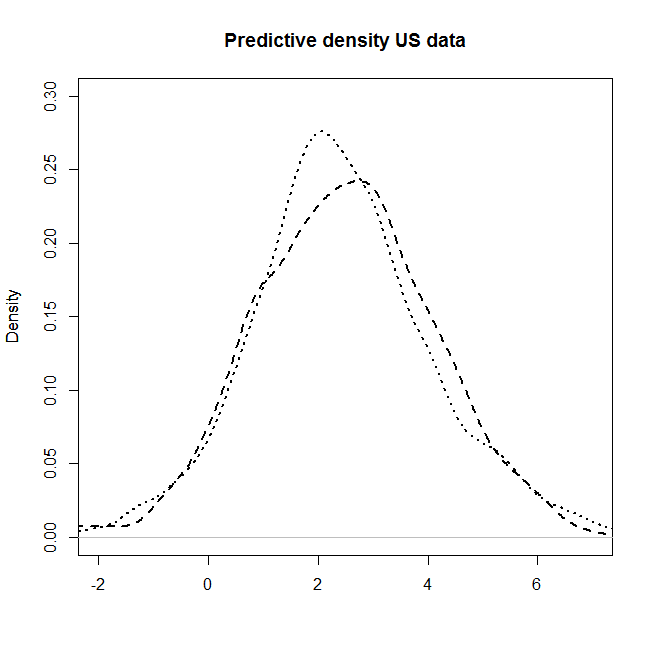}  
\caption{Real Data (dashed line) vs Posterior Predictive Distribution (dotted line) of the US Loss data. \label{fig:predUS}}
\end{figure}

\section{Discussion}\label{sc_disc}
An important aspect in analysis insurance loss data is their tendency to follow a skewed distribution and, because extreme events are not uncommon, to exhibit heavy tails. However, in some circumstances, symmetrical models with either non-heavy tails, like for example the normal distribution, or heavy tails, such as the $t$ density, may be effectively employed. Usually, this type of modelling can be achieved by considering the data in the log-scale. Although competing models could be estimated and assessed for their inferential and predictive performances, it is appealing to be able to consider a single model which, on the basis of the estimated values of some of the parameters, ``adjusts'' itself, in a sort of an automated fashion, to the problem under consideration. The asymmetrical Student-$t$ considered in this paper to represent insurance loss data has the above flavour. In fact, the estimated value of the skewness parameter $\alpha$ would suggest a symmetrical or non-symmetrical scenario, and the number of degrees of freedom would indicate the fatness of the tails of the data.

When dealing with a $t$ distribution, whether the usual symmetrical one or the one considered here, the estimation of the number of degrees of freedom has always been a challenge. It is not uncommon that the problem is somehow eluded by either setting $\nu$ on the basis of some appropriate theoretical results, or by consider different values and chose the most appropriate on the basis of some criteria. A Bayesian approach, and in particular an objective Bayesian approach, allows to obtain reliable estimates of the number of degrees of freedom with minimal initial input. We have here consider the number of degrees of freedom discrete and bounded above on the basis of the well known property of the $t$ density to converge to the normal distribution for sufficiently large $\nu$. With this consideration, we have been able to consider the objective prior for $\nu$ presented in \cite{Villa:Walker:2014a}. In addition, we prove here that the prior is not dependent on the value of the skewness parameter $\alpha$, therefore directly applicable to the model here considered.

We have studied the frequentist properties of the posterior for $\nu$ yielded by the truncated objective prior. As expected from the analytical result discussed in Section \ref{sc_priors}, different values of the skewness parameter do not affect the performance of the prior, except regular small variations due to the randomness of the repeating sampling. 
As expected, the performance of the prior distribution, in terms of MSE, improves when the sample size increases. The above results are in line with the ones obtained in \cite{Villa:Walker:2014a}.

We have then employed the discussed model to analyse real insurance loss data. The work has been carried out by using objective priors for each parameter of the model, to ensure a minimal information approach to the problem. In the first illustration we look to a well-known data set of losses due to fire in Denmark. The peculiarity of the data is to show strong (positive) skewness and to have extreme values. Even by considering the logarithmic transformation of the observations skewness and extremeness cannot be removed. The inferential procedure shows the model adjustment to the scenario by resulting in an $\alpha$ very close to zero and a number of degrees of freedom equal to 9.

The same Bayesian model, i.e. sampling distribution and prior distribution, is applied to a different data set. This data set, which contains indemnity losses in the US market, shows skewness and extreme values as well, but these appear to be removed once the logarithm of the observations is considered. A posterior mean of $\alpha=0.5$ indicates a symmetrical distribution, and a posterior median of $\nu=28$ indicates a distribution that is not very different from a normal density.

The above two results show how the same model can be applied to insurance loss data sets, without the need to change neither any of of its components nor the prior distributions for the parameters.

To conclude, it is obvious that the same approach here illustrated can be adopted to other types of data, which exhibit similar characteristics of skewness. That is, the objective Bayesian analysis of data by means of the AST model, including both the distribution and the priors for the parameters, can be generalised and employed in other disciplines, such as finance, environmental sciences and engineering, for example.

\section*{Acknowledgements}
The authors are thankful to the Associate Editor and the anonymous reviewer for their useful comments which significantly improved the quality of the paper. The authors are very grateful to F. J. Rubio for all the stimulating discussions and suggestions. Fabrizio Leisen's research has been supported by the European Community's Seventh  Framework Programme[FP7/2007-2013] under grant agreement no: 630677.


\renewcommand{\theequation}{A-\arabic{equation}}
\setcounter{equation}{0}  
\section*{APPENDIX - Monte Carlo Algorithm}\label{sc_app1}  

To sample from the marginal posterior distribution for each parameter, as they were analytically intractable, we had to use Monte Carlo methods. In particular, we have implemented a Gibbs sampler where we define a Metropolis-Hastings proposal at each step for every parameter. The details of the algorithm are given below.

\noindent
At a given iteration $s$ the parameters are updated as follows.

\begin{enumerate}
\item[1.] Parameter $\nu$.

The proposal transition kernel for $\nu $ is a discrete uniform distribution (DU) defined between 1 and 30
\begin{equation*}
\nu ^{\ast }\sim \mathrm{DU}(1,30).
\end{equation*}
Thus $\nu ^{(s+1)}=\nu ^{\ast }$ with probability $\alpha _{\nu },$ where
\begin{equation*}
\alpha _{\nu }=\min \left[ 1,\frac{\pi \left( \Theta ^{\ast }|\mathbf{x}\right) }{\pi \left( \widetilde{\Theta }|\mathbf{x}\right) }\right],
\end{equation*}%
with $\Theta ^{\ast }=(\nu ^{\ast },\mu ^{(s)},\sigma ^{(s)},\alpha ^{(s)})$ and with $\widetilde{\Theta }=(\nu ^{(s)},\mu ^{(s)},\sigma ^{(s)},\alpha^{(s)}). $

\item[2.] Parameter $\mu$.

The proposal transition kernel for $\mu $ is a normal distribution with mean $\mu^{(s)}$ and variance $s_\mu$
\begin{equation*}
\mu ^{\ast }|\mu ^{(s)}\sim N(\mu ^{(s)},s_{\mu }),
\end{equation*}%
where $s_{\mu }$ is fixed such that the mixing of the chains is optimal.
Thus, $\mu ^{(s+1)}=\mu ^{\ast }$ with probability $\alpha _{\mu },$ where
\begin{equation*}
\alpha _{\mu }=\min \left[ 1,\frac{\pi \left( \Theta ^{\ast }|\mathbf{x}\right) \phi (\mu ^{(s)}|\mu ^{\ast },s_{\mu })}{\pi \left( \widetilde{\Theta }|\mathbf{x}\right) \phi (\mu ^{\ast }|\mu ^{(s)},s_{\mu })}\right],
\end{equation*}%
with $\Theta ^{\ast }=(\nu ^{(s)},\mu ^{\ast },\sigma ^{(s)},\alpha ^{(s)})$ and with $\widetilde{\Theta }=(\nu ^{(s)},\mu ^{(s)},\sigma ^{(s)},\alpha ^{(s)}).
$

\item[3.] Parameter $\sigma$.

The proposal transition kernel for $\sigma $ is given by a gamma distribution with parameters $a_\sigma$ and $b_\sigma$
\begin{equation*}
\sigma ^{\ast }|\sigma ^{(s)}\sim \mathrm{Ga}(a_{\sigma },b_{\sigma }),
\end{equation*}%
where $a_{\sigma }$ and $b_{\sigma }$ are appropriately chosen to obtain optimal mixing of the chains.
Thus, $\sigma ^{(s+1)}=\sigma ^{\ast }$ with probability $\alpha _{\sigma },$ where
\begin{equation*}
\alpha _{\mu }=\min \left[ 1,\frac{\pi \left( \Theta ^{\ast }|\mathbf{x}\right) f_{G}(\sigma ^{(s)}|a_{\sigma },b_{\sigma })}{\pi \left( \widetilde{\Theta }|\mathbf{x}\right) f_{G}(\sigma ^{\ast }|a_{\sigma },b_{\sigma })}\right] ,
\end{equation*}%
with $\Theta ^{\ast }=(\nu ^{(s)},\mu ^{(s)},\sigma ^{\ast },\alpha ^{(s)})$ and with $\widetilde{\Theta }=(\nu ^{(s)},\mu ^{(s)},\sigma ^{(s)},\alpha ^{(s)}).
$

\item[4.] Parameter $\alpha$.

The proposal transition kernel for $\alpha $ is given by a beta distribution with parameters $a_\alpha^{(s)}$ and $b_\alpha^{(s)}$
\begin{equation*}
\alpha ^{\ast }|\alpha ^{(s)}\sim \mathrm{Be}(a_{\alpha }^{(s)},b_{\alpha}^{(s)}),
\end{equation*}%
where $a_{\alpha }^{(s)}$ and $b_{\alpha }^{(s)}$ are chosen in order that the mean of the beta distribution is $\alpha ^{(s)}$. Thus, if we indicate by $V_\alpha$ the variance of the beta kernel, we have 
\begin{eqnarray*}
a_{\alpha }^{(s)} &=&\left( ( 1-\alpha ^{(s)}) /V_{\alpha}-1/\alpha ^{(s)}\right) \cdot \left( \alpha ^{(s)}\right) ^{2} \\
b_{\alpha }^{(s)} &=&a_{\alpha }\left( 1/\alpha ^{(s)}-1\right),
\end{eqnarray*}%
with $V_{\alpha }$ is selected to obtain optimal mixing of the chains.
Thus, $\alpha ^{(s+1)}=\alpha ^{\ast }$ with probability $\alpha _{\alpha },$ where
\begin{equation*}
\alpha _{\alpha }=\min \left[ 1,\frac{\pi \left( \Theta ^{\ast }|\mathbf{x}\right) f_{Be}(\alpha ^{(s)}|\alpha ^{\ast },V_{\alpha })}{\pi \left(\widetilde{\Theta }|\mathbf{x}\right) f_{Be}(\alpha ^{\ast }|\alpha^{(s)},V_{\alpha })}\right],
\end{equation*}%
with $\Theta^{\ast }=(\nu ^{(s)},\mu ^{(s)},\sigma ^{(s)},\alpha ^{\ast })$ and with $\widetilde{\Theta }=(\nu ^{(s)},\mu ^{(s)},\sigma ^{(s)},\alpha ^{(s)}).$
\end{enumerate}

\end{document}